\documentclass[11pt]{article}
\usepackage{amsfonts,amssymb,amsmath,amsthm,bbm,mathrsfs}
\usepackage{graphicx, color, psfrag,hyperref,xcolor,bbm}
\usepackage{mathrsfs} 
\usepackage{tikz-cd}
\usepackage{comment}
\usepackage{hyperref}
\usepackage[english]{babel}
\allowdisplaybreaks

\usepackage{color}

\newcommand{\ew}{\color{black}}

\numberwithin{table}{section}
\newtheorem{theorem}{Theorem}[section]
\newtheorem{lemma}[theorem]{Lemma}
\newtheorem{corollary}[theorem]{Corollary}
\newtheorem{proposition}[theorem]{Proposition}

\newtheorem{definition}[theorem]{Definition}

\def\exp{\mathop{\textrm{\rm exp}}\nolimits}              

\newcommand\blue[1]{\textcolor{blue}{}}

\newcommand{\be}{\begin{equation}}
\newcommand{\ee}{\end{equation}}

\newcommand{\trees}{\mathcal{T}}

\newcommand{\card}[1]{\left| #1 \right|}
\newcommand{\bbd}[1]{\boldsymbol{#1}}
\newcommand{\bbz}{\mathbb Z}
\newcommand{\bbr}{\mathbb R}

\newcommand{\La}{\ensuremath{\Lambda}}
\newcommand{\si}{\ensuremath{\sigma}}

\newcommand*\diff{\mathop{}\!\mathrm{d}}

\newcommand{\pro}{\mathcal{P}}

\newcommand{\B}{\mathcal{B}}

\newcommand{\empbf}[1]{{\bf \emph{#1}}}

\begin{document}

\title{{\bf Local Central Limit Theorem for unbounded long-range potentials}}
 
 \author{Eric O. Endo 
 \ew 
 \footnote{NYU-ECNU Institute of Mathematical Sciences at NYU Shanghai, 3663 Zhongshan Road North, Shanghai, 200062, China
 \newline
 email: ericossamiendo@gmail.com}\\
 Roberto Fern\'andez \footnote{ NYU-ECNU Institute of Mathematical Sciences at NYU Shanghai, 3663 Zhongshan Road North, Shanghai, 200062, China
 \newline
 email: rf87@nyu.edu }\\ 
  Vlad Margarint \footnote{University of North Carolina at Charlotte, USA
 \newline
 email: vmargari@uncc.edu }\\ 
  Tong Xuan Nguyen \footnote{NYU-ECNU Institute of Mathematical Sciences at NYU Shanghai, 3663 Zhongshan Road North, Shanghai, 200062, China
 \newline
 email:  tn2137@nyu.edu }
 }

\maketitle

\begin{center}
{\bf Abstract} 
\end{center}

We prove the equivalence between the integral central limit theorem and the local central limit theorem for two-body potentials with long-range interactions on the lattice $\mathbb{Z}^d$ for $d\ge 1$.  The spin space can be an arbitrary, possibly unbounded subset of the real axis with a suitable a-priori measure. For general unbounded spins, our method works at high-enough temperature, but for bounded spins our results hold for every temperature.  Our proof relies on the control of the integrated characteristic function, which is achieved by dividing the integration into  three different regions, following a standard approach proposed forty years ago by Campanino, Del Grosso and Tirozzi. The bounds required in the different regions are obtained through cluster-expansion techniques.  For bounded spins, the arbitrariness of the temperature is achieved through a decimation (``dilution") technique, also introduced in the later reference.

\vspace{2cm}

{\footnotesize
 {\em  AMS 2000 subject classification}: Primary-  ; secondary- 

}

\section{Introduction}

The study of the equivalence between the integral central limit theorem and the local central limit theorem for spin systems  has a long and illustrious history \cite{CCT,CDT,Del Grosso,DT,PS}.  Its interest stems mainly in relation with the issue of  equivalence of the ensembles. Previous references, however, included different assumptions and restrictions that were largely lifted in two recent references. The work in \cite{CDT,Del Grosso,DT} considered finite-range potentials. While \cite{Del Grosso,DT} concentrated on finite and compact spins, \cite{CDT} focused on spins on the set of integer or real numbers; In \cite{CCT} the results was generalized to a family of finite-spins long-range potentials with technical constraint; For non-zero external field, see \cite{IS}; In \cite{EM} some of us recently proved the equivalence for all absolutely summable and translation-invariant potentials at sufficiently large temperature, removing the constraint in \cite{CCT}.
 For bounded spins, the result is further extended to \emph{all} temperatures in the recent work \cite{PS} where also the condition of the translation-invariance of the potential is also removed.  The remaining limitation refers to the structure of the spin space:  Reference \cite{PS} applies to finite discrete-spin spaces, while the work in \cite{EM} assumes absolutely summability of the interaction, which excludes Ising-type interactions with real-valued spins.  

In the present work we prove the equivalence for general, possibly unbounded spins with long-range two-body interactions without assuming translation invariance of the potential.  Following Procacci's and Scoppola's \cite{PS}, using their ``dilution" technique (explained below) we extend the results to arbitrary temperatures when the spins are bounded.  The reader will also recognize some common (natural) combinatorial bounds employed by that reference.  Nevertheless we resort to a slightly different cluster expansion that leds, in our opinion, to a slightly more efficient treatment.  We believe our results to be the most general possible in the standard statistical mechanical framework of two-body spin interactions.

\section{Set up}

\subsection{Continuous spin systems}

In this section we introduce the model we study. While our combinatorial approach is common to all spin spaces, relevant bounds lead to a distinction between the case of bounded and unbounded spins. Our exposition will address mainly the more involved unbounded case, indicating in each case the relevant particularities to the bounded case.  For unbounded spins, the Hamiltonian will be defined following \cite{COPP}, where the authors give criteria for the boundary conditions to guarantee the equivalence between the thermodynamic limits of finite volume Gibbs measures with external boundary and measures satisfying the DLR equations. Here, we consider a slightly more restricted set of boundary conditions.  Of course, no restriction is needed for the boundary conditions of the bounded-spin case. Here are the ingredients of our setup.

\paragraph{Configuration space:}
We consider the set $\Omega = E^{\bbz^d}$, where $E=\bbr$ for the unbounded-spin case, and $E\subset[-R,R]$, for some $R>0$, for the bounded case and $d\ge 1$. Configurations are denoted with greek letters, $\sigma=(\sigma_x)_{x\in \mathbb{Z}^d}\in \Omega$, with the spins $\sigma_x$ being real numbers. The set $\Omega$ is a topological space with product topology inherited from $\bbr$, it is a Polish space, i.e., it is metrizable with a metric for which it is separable and complete.  

\paragraph{A-priori single-spin measure:}
The single-spin space $E$ is equipped with a Borel measure, denoted  $\mathrm{d}\sigma_x$ for each spin $\si_x$ with $x\in \mathbb{Z}^d$.  In our exposition we assume that this measure is the usual Lebesgue measure in both the bounded and unbounded case.  With minor changes (needed only in Section \ref{sec:hard} due to the use of integration by parts), our proofs apply also to the discrete case, in which $d\sigma_x$ is the counting measure.

The \emph{a-priori} measure $\nu$ on $S$ incorporates an exponential weight and takes the form  
\[
\nu(\diff \sigma_x)=e^{-F(\si_x)}  \diff \sigma_x\;.
\]
where $F:\bbr \to \bbr$ is ``sufficiently well behaved''.  
For the bounded case, a possible choice is $F=0$, but it can also incorporate single-spin terms of the interaction, such as magnetic-field terms. 
For the proofs below we need the following properties:
\begin{itemize}
\item $\nu$ is a finite measure.
\item $F$ is differentiable except in finitely many points (we use integration by parts!).
\item $\card{F'}$ has sub-Gaussian behavior at infinity (trivially true in the bounded case).
\end{itemize}

\paragraph{Interactions and Hamiltonians:}
We assume the two-body interactions of the form $J_{xy}\sigma_x\sigma_y$ with $x\neq y\in \mathbb{Z}^d$ such that the couplings $J_{xy}$ are
\begin{itemize}
\item symmetric: $J_{xy}=J_{yx}$ and
\item absolutely summable: 
\be\label{abssum}
\sum_{y\in \bbz^d}|J_{xy}|<\infty \text{ for every $x$.}
\ee
\end{itemize}

For a given finite set $\Lambda \Subset \mathbb{Z}^d$, consider $\Omega_{\Lambda}=\mathbb{R}^{\Lambda}$ be the set of configurations in $\Lambda$. The Hamiltonian with free boundary condition on $\Omega_{\Lambda}$ is defined by
\[
H_{\Lambda}(\sigma)=-\sum_{x,y\in \Lambda}J_{xy}\sigma_x\sigma_y,
\]
and $H_{\Lambda}(\sigma)=0$ if $|\Lambda|=1$.

\paragraph{Stability assumption:} This requirement is superfluous for the bounded-spin case.  Our approach requires good integrability properties of the combined exponential weight.
\begin{definition}[Superstability]\label{superstability}
There exist $A>0$ and $c\in \mathbb{R}$ such that, for every $\Lambda\Subset \mathbb{Z}^d$, $\beta>0$ and  $\sigma_{\Lambda}\in \mathbb{R}^{\Lambda}$,
\be\label{eq:superstability}
\sum_{x\in \Lambda}F(\sigma_x)+\beta H_{\Lambda}(\sigma) \ge \sum_{x\in \Lambda}(A\sigma_x^2-c).
\ee
\end{definition}
Note that, for $\Lambda=\{x\}$, the superstability implies 
\be
F(\sigma_x)\ge A\sigma^2_x -c
\ee
for every $x\in \mathbb{Z}^d$, so the a priori measure is bounded by a Gaussian.  In particular this implies that functions of the form $\card{\sigma_x}^\ell$, $\ell>0$; $e^{c\card{\sigma_x}}$ ---and the same multiplied by $\card{F'}$--- are all $\nu$-integrable.

\paragraph{Boundary conditions:} In the unbounded case, it is well known that not every boundary condition can be allowed in order to construct a physically meaningful statistical mechanics. 

\begin{definition}[Strongly tempered boundary condition]\label{tempered}
We say that a boundary condition $\omega \in \Omega$ is \emph{strongly tempered} if
\be\label{eq:temp}
\sup_{\Lambda \Subset \mathbb{Z}^d}\sup_{x\in \Lambda}\sum_{y\in \Lambda^c}|J_{xy} \omega_y|\;=:\;\vartheta(\omega) <\infty
\ee
and we denote the set of strong tempered boundary condition by $\Omega^{\text{temp}}$. 
\end{definition}
This condition is more restrictive than the ones considered in \cite{COPP,MN}, but it is general enough and easier to manipulate.  
In particular, we have added a ``sup'' in \eqref{eq:temp} because the proof of our main result requires the function $\vartheta$ not to depend on $\Lambda$.

In the bounded case, \emph{all} boundary conditions are acceptable, and we have
\be
\sup_{\Lambda \Subset \mathbb{Z}^d}\sup_{x\in \Lambda}\sum_{y\in \Lambda^c}|J_{xy} \omega_y|
\;\le\; R \sum_{y\in \Lambda^c}|J_{xy}| \;=:\; \vartheta
\ee

\paragraph{Finite-volume Gibbs measures:}
The Hamiltonian at finite volume $\Lambda$ with strong tempered boundary condition $\omega$ is given by
\[
H^{\omega}_{\Lambda}(\si)= -\sum_{x,y\in \Lambda}J_{xy}\sigma_x\sigma_y-\sum_{\substack{x\in \Lambda \\ y\in \La^c}} J_{xy}\si_x\omega_y.
\]
The Gibbs measure at finite volume $\Lambda$ with  strong tempered boundary condition $\omega$ is defined on $\Omega_{\Lambda}$ by
\be\label{gibbsmeasure}
\mu^{\omega}_{\Lambda,\beta}(\diff \sigma)= \frac{1}{Z^{\omega}_{\Lambda,\beta}} e^{-\beta H^{\omega}_{\Lambda}(\sigma)}\nu(\diff \si_{\Lambda}),
\ee
where, here, we are  considering $\nu$ to be the product measure
\[
\nu(\diff \si_{\Lambda})=\prod_{x\in \La}\nu(\diff \si_x ),
\]
and the partition function is given by
\[
Z^{\omega}_{\La,\beta}=\int_{\Omega_{\La}}e^{-\beta H^{\omega}_{\Lambda}(\sigma)}\nu(\diff \si_{\La}).
\]


\subsection{The Integral and the Local Central Limit Theorems}

Consider the finite cube $\Lambda_k\Subset \mathbb{Z}^d$ defined by $\Lambda_k=[-k,k]^d$. Define, for a given configuration $\sigma\in \Omega_{\Lambda_k}$,
\be
S_k = \sum_{x\in \Lambda_k}\sigma_x \quad \text{ and }\quad \bar{S}_k=\frac{S_k-\mu^\omega_{\Lambda_k,\beta}(S_k)}{\sqrt{D_k}},
\ee
where $D_k=\mu^\omega_{\Lambda_k,\beta}((S_k-\mu^\omega_{\Lambda_k,\beta}(S_k))^2)$ denotes the variance of $S_k$ under the Gibbs measure $\mu^\omega_{\Lambda_k,\beta}$.

\begin{definition}
Let $(\Lambda_k)_{k\ge 1}$ be a sequence of finite set of cubes in $\mathbb{Z}^d$ and $\omega$ be a strong tempered boundary condition. We say that a sequence of Gibbs measures $(\mu^{\omega}_{\Lambda_k,\beta})_{k\ge 1}$ satisfies the \emph{integral central limit theorem} if the following three conditions are satisfied:
\begin{enumerate}
\item[(i)] $\displaystyle\lim_{k\to \infty}D_k/|\Lambda_k|=L$,
\item[(ii)] $L>0$,
\item[(iii)] For every $\tau\in \mathbb{R}$,
\be
\lim_{k\to \infty}\mu^{\omega}_{\Lambda_k,\beta}(\bar{S}_k\le \tau) = \frac{1}{\sqrt{2\pi}}\int_{-\infty}^{\tau} e^{-z^2/2}\diff z.
\ee
\end{enumerate}
\end{definition}

To define the local central limit theorem, we need to introduce the probability density $p^{\omega}_{\Lambda,\beta}$ given by
\begin{equation}
\mu^{\omega}_{\Lambda_k,\beta}(\bar{S}_k \in A) = \int_A p^{\omega}_{\Lambda_k,\beta}(x)\diff x,
\end{equation}
where $A$ is an element in a Borel set of $\mathbb{R}$. Note that, by the Fourier inversion formula,
\be \label{eq:fourier}
p^{\omega}_{\Lambda_k,\beta}(x) = \frac{1}{2\pi} \int_{-\infty}^{\infty}\mu^{\omega}_{\Lambda_k,\beta}(e^{it\bar{S}_k}) \,e^{itx}\diff t.
\ee

 \begin{definition}\label{def:lclt}
 Let $(\Lambda_k)_{k\ge 1}$ be a sequence of finite set of cubes in $\mathbb{Z}^d$ and $\omega$ be a strong tempered boundary condition. 
We say that a sequence of Gibbs measures $(\mu^{\omega}_{\Lambda_k,\beta})_{k\ge 1}$ satisfies the \emph{local central limit theorem} if the following three conditions are satisfied:
\begin{enumerate}
\item[(i)] $\displaystyle\lim_{k\to \infty}D_k/|\Lambda_k|=L$,
\item[(ii)] $L>0$,
\item[(iii)] The probability density $p^{\omega}_{\Lambda_k,\beta}$ satisfies
\be \label{eq:lclt}
\lim_{k\to \infty}\sup_{x\in \mathbb{R}} \left|p^{\omega}_{\Lambda_k,\beta}(x) -\frac{1}{\sqrt{2\pi}}e^{-x^2/2}  \right|=0.
\ee
\end{enumerate}
\end{definition}

\section{Main result and guidelines of its proof}
The following is our main result.

\begin{theorem}\label{main}
Let $(\Lambda_k)_{k\ge 1}$ be a sequence of finite set of cubes in $\mathbb{Z}^d$ for $d\ge 1$ and $\beta>0$. Assume that (i) the spins are bounded and $\beta>0$ (ii) the spins are unbounded, the boundary condition $\omega$ is strongly tempered, and $\beta(\omega)$ is sufficiently small. If a sequence of Gibbs measures $(\mu^{\omega}_{\Lambda_k,\beta})_{k \ge 1}$ satisfies the Integral Central Limit Theorem, then that sequence satisfies the Local Central Limit Theorem
\end{theorem}

It is simple to see that the validity of the local CLT implies the validity of the integrated one.  Therefore the previous theorem proves equivalence between both theorems.
The proof is based on the following guidelines.

\paragraph{$L^1$ convergence of the characteristic function:}
Due to \eqref{eq:fourier},
\be
p^{\omega}_{\Lambda_k,\beta}(x) -\frac{1}{\sqrt{2\pi}} e^{-x^2/2} 
= \frac{1}{2\pi}\int \left[ \mu^{\omega}_{\Lambda_k,\beta}(e^{it\bar{S}_k}) - e^{-t^2/2} \right]\, e^{itx}\diff t\;.
\ee
Thus, a sufficient condition to prove (iii) of the local central limit theorem is
\be \label{eq:sufficient}
\lim_{k\to \infty} \int \left| \mu^{\omega}_{\Lambda_k,\beta}(e^{it\bar{S}_k}) - e^{-t^2/2} \right|\diff t\;=\; 0\;.
\ee
We see that, while the validity of the Central Limit Theorem is equivalent to the pointwise convergence 
\be \label{eq:cltr}
\lim_{k\to \infty} \left| \mu^{\omega}_{\Lambda_k,\beta}(e^{it\bar{S}_k}) - e^{-t^2/2} \right|\;=\; 0\;,
\ee
the local central limit follows from the corresponding $L^1$ convergence.

\paragraph{Decomposition into infrared and ultraviolet terms:} 
The passage from pointwise to $L^1$ convergence amounts to controlling the integral for large values of the ``frequency" $t$.  Therefore is natural to resort to a decomposition of the form
\be\label{eq:decomposition-1}
T_k\;:=\; \int \left| \mu^{\omega}_{\Lambda_k,\beta}(e^{it\bar{S}_k}) - e^{-t^2/2} \right|\diff t\;=\; T_k^{\le B}+ T_k^{>B}
\ee
with
\be
\begin{array}{c} T_k^{\le B}\\ T_k^{>B} \end{array} \;=\;
\begin{array}{c} \int_{-B}^B\\ \int_{[-B,B]^c} \end{array} 
\left| \mu^{\omega}_{\Lambda_k,\beta}(e^{it\bar{S}_k}) - e^{-t^2/2} \right|\diff t\;.
\ee
The validity of the central limit theorem implies that $T_k^{\le B}\to 0$ with $k$, whichever $B$, so we have to take care of the ``ultraviolet" part  $T_k^{> B}$.  In turns, the latter will be controlled through the inequality
\be\label{eq:further}
T_k^{>B} \;\le\; \int_{|t|\ge B} e^{-t^2/2}\diff t + \int_{|t|\ge B} \left|\mu^{\omega}_{\Lambda_k,\beta}(e^{it\bar{S}_k})\right| \diff t\;.
\ee
The first term on the right can be made arbitrarily small by choosing $B$ sufficiently large.  Hence, the proof of \eqref{eq:lclt} boils down to proving that the second term converges to zero as $k\to\infty$ for (sufficiently large) $B>0$.  In fact we will prove that
\be\label{boildown}
\lim_{k\to \infty}  \int \left|\mu^{\omega}_{\Lambda_k,\beta}(e^{it\bar{S}_k})\right| \diff t\;=\;0\;.
\ee

\paragraph{Further decomposition into soft, medium and hard ultraviolet terms:} 
The definition of $\bar{S}_k$ shows that the value of $\mu^{\omega}_{\Lambda_k,\beta}(e^{it\bar{S}_k})$ depends on the scaled frequencies $t/\sqrt{D_k}$.  The necessary bounds require to split the analysis in three frequency domains.  They are the following:
\begin{itemize}
\item \emph{Soft frequencies:} Defined by the condition $B\le |t|\le \delta \sqrt{D_k}$ for $\delta$ sufficiently small.  In this domain $\left|\mu^{\omega}_{\Lambda_k,\beta}(e^{it\bar{S}_k})\right|$ is controlled by the lowest order of $e^{it\bar{S}_k}$ leading to non-trivial contributions (second order).
\item \emph{Hard frequencies:} Defined by the condition $|t|>T\sqrt{D_k}$ for $T$ sufficiently large.  In this domain $\left|\mu^{\omega}_{\Lambda_k,\beta}(e^{it\bar{S}_k})\right|$ is controlled by the exponential factors $F(\sigma_x)$. 
This term is absent if spins are either discrete and bounded or lattice-distributed, as was the case in previous work \cite{CDT,EM, PS}  
\item \emph{Medium frequencies:} Defined by the condition $\delta\sqrt{D_k} \le |t|\le T\sqrt{D_k}$ wih $\delta$ and $T$ determined above.  The control of $\left|\mu^{\omega}_{\Lambda_k,\beta}(e^{it\bar{S}_k})\right|$ is a consequence of some trigonometric magic.
\end{itemize}
The first two regions were already introduced in \cite{CDT} and used in several subsequent articles \cite{CCT, EM, PS}.  The last region is a novel contribution of our work.  

Formally we decompose
\be\label{eq:dec}
\int_{|t|\ge B} \left|\mu^{\omega}_{\Lambda_k,\beta}(e^{it\bar{S}_k})\right| \diff t\;=\; S_k\bigl(\{B\le |t|\le \delta \sqrt{D_k}\}\bigr) +
    S_k\bigl(\{\delta\sqrt{D_k} \le |t|\le T\sqrt{D_k}\}\bigr) + S_k\bigl(|t|\ge T\sqrt{D_k}\}\bigr) 
\ee
with
\be\label{eq:dec-a}
S_k(A)\;=\; \int_A \left|\mu^{\omega}_{\Lambda_k,\beta}(e^{it\bar{S}_k})\right| \diff t\;.
\ee
Our proof will show:
\begin{itemize}
\item[(a)] For $\delta$ is sufficiently small, $S_k\bigl(\{B\le |t|\le \delta \sqrt{D_k}\}\bigr) \le H_\delta(B)$, uniformly in $k$,  with $H_\delta(B) \to 0$ as $B\to \infty$.
\item[(b)] For $T$ sufficiently large, $S_k\bigl(|t|\ge T\sqrt{D_k}\}\bigr) \to 0$ as $k\to \infty$.
\item[(c)]  $S_k\bigl(\{\delta\sqrt{D_k} \le |t|\le T\sqrt{D_k}\}\bigr) \to 0$ as $k\to\infty$ for all $0<\delta<T$,.
\end{itemize}

\paragraph{Lattice dilution:} The last ingredient of our toolbox, taken from \cite{PS}, is the passing to diluted (scaled) lattices $\mathbb{Z}^d_{r_0} $, 
for integers $r_0\ge 1$, defined as
\be
\mathbb{Z}^d_{r_0} = \{(r_0 x_1,r_0 x_2,\ldots,r_0 x_d): (x_1,\ldots,x_d)\in \mathbb{Z}^d\}\,
\ee
that is, to lattices where nearest-neighbor sites are at a distance $r_0$. Dilution has a double advantage:
\smallskip\par\noindent
\emph{{\bf 1)} Expectations in the original lattice can be bounded by the (simpler) expectations in the diluted lattice}.  Indeed, 
defining $\tilde{\Lambda}_k =\tilde{\Lambda}_k(r_0) = \Lambda_k \cap \mathbb{Z}^d_{r_0}$ we have, by spatial Markov property,
\begin{align*}
\left| \mu^{\omega}_{\Lambda_k,\beta}(e^{it\bar{S}_k})\right|
&=\left|\mu^{\omega}_{\Lambda_k,\beta}\left(\exp\left( \frac{it}{\sqrt{D_k}}S_k \right) \right)\right| \\
&= \left|  \mu^{\omega}_{\Lambda_k,\beta}\left( \mu^{\omega}_{\Lambda_k,\beta} \left( \exp\left(\frac{it}{\sqrt{D_k}}S_k\right) \Bigg| \sigma_{\Lambda_k\setminus \tilde{\Lambda}_k}=\omega'_{\Lambda_k\setminus \tilde{\Lambda}_k}\right) \right)\right|\\
&\le \sup_{\omega'\in \Omega^{\text{temp}}} \left| \mu^{\omega}_{\Lambda_k,\beta} \left( \exp\left(\frac{it}{\sqrt{D_k}}\sum_{x\in \tilde{\Lambda}_k}\sigma_x\right) \Bigg| \sigma_{\Lambda_k\setminus \tilde{\Lambda}_k}=\omega'_{\Lambda_k\setminus \tilde{\Lambda}_k}\right) \right|\\
&= \sup_{\omega'\in \Omega^{\text{temp}}} \left|\mu^{(\omega\lor \omega')^{r_0}}_{\tilde{\Lambda}_k,\beta}\left(\exp\left( \frac{it}{\sqrt{D_k}}\sum_{x\in \tilde{\Lambda}_k}\sigma_x \right)  \right)\right|,
\end{align*}
where the configuration $(\omega\lor \omega')^{r_0}$ is defined by
$$
(\omega\lor \omega')^{r_0}(x) = 
\begin{cases}
\omega'(x) & \text{ if }x\in \Lambda_k \setminus \tilde{\Lambda}_k\\
\omega(x) & \text{ if }x\in \Lambda^c_k
\end{cases}.
$$
Note that $(\omega\lor \omega')^{r_0}$ is strong tempered since both $\omega$ and $\omega'$ are strong tempered. Thus, we can bound the characteristic function of $\bar{S}_k$ by
\be\label{eq:bound-rfr}
\left| \mu^{\omega}_{\Lambda_k,\beta}(e^{it\bar{S}_k})\right| \le  \sup_{\omega\in \Omega^{\text{temp}}}\left|\mu^{\omega}_{\tilde{\Lambda}_k,\beta}\left(\exp\left( \frac{it}{\sqrt{D_k}}\sum_{x\in \tilde{\Lambda}_k}\sigma_x \right)  \right)\right|.
\ee
\smallskip\par\noindent
\emph{{\bf 2)} In the bounded case, dilution allows us to apply high-temperature techniques for arbitrary temperatures.}
Indeed, the diluted expectation in the right-hand side of \eqref{eq:bound-rfr} corresponds to spins subjected to an effective temperature
\be
\widehat{\beta}_{r_0}:=\beta\sup_{x\in \mathbb{Z}^d_{r_0}}\sum_{y\in\mathbb{Z}^d_{r_0}}|J_{xy}|.
\ee
As the scaling factor $r_0$ grows, the value of $\widehat{\beta}_{r_0}$ decreases because of the absolute summability of the two-body interaction.  Hence a sufficiently large dilution $r_0$ brings a system with any $\beta>0$ into a system with arbitrarily high temperature and, hence, allows the application of high-temperature cluster expansions.

\paragraph{No dilution in the unbounded case:} 
In the unbounded case, the dependence on boundary conditions in our bounds may lead to a diverging supremum in \eqref{eq:bound-rfr}.  This makes the dilution technique not applicable in this case, and in all the expressions below we must take $r_0=1$ and, thus, $\tilde{\Lambda}_k=\Lambda_k$.

\section{Proof of Theorem \ref{main}. (A) Basic expressions}\label{sec.proof-a}

\subsection{Partition functions in terms of polymers}
Our basic expression is
\be\label{ratio}
\mu^{\omega}_{\tilde{\Lambda}_k,\beta}\left(\exp\left( \frac{it}{\sqrt{D_k}}\sum_{x\in \tilde{\Lambda}_k}\sigma_x \right)  \right)\;=\;
\frac{Z^{\omega}_{\tilde{\Lambda}_k,\beta}(t)}{Z^{\omega}_{\tilde{\Lambda}_k,\beta}(0)}
\ee
where
\begin{equation}
    Z^{\omega}_{\tilde{\Lambda}_k,\beta}(t) =\int_{\mathbb R^{|\tilde{\Lambda}_k|}}\exp\left(-\beta H^{\omega}_{\tilde{\Lambda}_k}(\sigma)+\frac{it}{\sqrt{D_k}}\sum_{x\in \tilde{\Lambda}_k}\sigma_x\right)\prod_{x\in\tilde{\Lambda}_k}\nu(\diff\sigma_x).
\end{equation}
Note that $Z^{\omega}_{\Lambda_k,\beta}(0)=Z^{\omega}_{\Lambda_k,\beta}$. This ratio of partition functions will be bounded through the (high-temperature) cluster expansion.  
As reviewed in Appendix \ref{sec:cluster-exp}, the starting step to set a cluster expansion is to expand the partition function in terms of products of non-intersecting (compatible) objects called \emph{polymers}.   For one and two-body interactions, the polymers are defined as follows.

Denote $\mathcal{P}_{1,2}$ be a family of non-empty subsets $X\Subset \mathbb{Z}^d$ consisting of at most two points, that is, $1\le |X|\le 2$. A \emph{polymer} $R$ is a finite set $\{X_1,\ldots,X_p\}$ of elements in $\mathcal{P}_{1,2}$ that are connected in the following sense: for any $X_l,X_m \in R$, there exist a sequence $X_{k_1},\ldots,X_{k_q}$ of elements in $R$ such that $X_{k_1}=X_l$, $X_{k_q}=X_m$, and $X_{k_j}\cap X_{k_{j+1}}\neq \emptyset$. Let $\mathcal{R}$ be the set of all polymers and, if $R\in \mathcal{R}$, denote by $\underline{R}$ the support of the polymer $R$ given by $\underline{R}=\bigcup_{X\in R}X$.  We then have
\begin{equation}\label{eq:Xuan268_1}
Z^{\omega}_{\tilde{\Lambda}_k,\beta}(t) = \left(\prod_{x\in \tilde{\Lambda}_k}\int_{\mathbb R}e^{-F^{\omega}_{\tilde{\Lambda}_k,\beta}(\sigma_x)}\diff \sigma_x\right)\left(1+\sum_{n=1}^\infty\sum_{(R_1,\ldots, R_n)\in\mathcal{R}^n}\prod_{1\le i<j\le n}\mathbbm{1}_{\underline{R}_i\cap\underline{R}_j=\emptyset}\prod_{i=1}^n\zeta^{(1)}_{t,r_0}(R)\right),
\end{equation}
where the activity function $\zeta^{(1)}_{t,r_0}$ is given by
\begin{equation}
\zeta^{(1)}_{t,r_0}(R)=\int_{\mathbb R^{|\underline{R}|}}\prod_{\{x\} \in R} \left(\exp\left(\frac{it\sigma_x}{\sqrt{D_k}}\right)-1\right)\prod_{\{x,y\}\in R}\left(e^{\beta J_{xy}\sigma_x\sigma_y}-1\right)\prod_{x\in\underline{R}}\nu^{\omega}_{\tilde{\Lambda}_k,\beta}(\diff\sigma_x)\;.
\end{equation}
The measures $\nu^{\omega}_{\Lambda,\beta}(\diff\sigma_x)$ are defined by the densities
\begin{equation}\label{nu}
\nu^{\omega}_{\Lambda,\beta}(\sigma_x):=\frac{\displaystyle e^{-F^{\omega}_{\Lambda,\beta}(\sigma_x)}}{\displaystyle\int_{\mathbb R}e^{-F^{\omega}_{\Lambda,\beta}(s)}\nu(\diff s)}
\end{equation}
where
\begin{equation}
F^{\omega}_{\Lambda,\beta}(\sigma_x)=F(\sigma_x)-\beta\sigma_x\sum_{y\in \Lambda^c}J_{xy}\omega_y\;.
\end{equation}
In the sequel we will find useful to consider the finite (non-normalized) measures
 $\pi^{\omega}_-$ and $\pi^{\omega}_+$ defined by 
\begin{equation}
\pi^{\omega}_-(\sigma_x):= \frac{\displaystyle e^{-F(\sigma_x)-\vartheta(\omega)\beta|\sigma_x|}}{\displaystyle\int_{\mathbb R}e^{-F(s)+\vartheta(\omega)\beta|s|}\diff s}
    \quad \text{and}  \quad 
\pi^{\omega}_+(\sigma_x):= \frac{\displaystyle e^{-F(\sigma_x)+\vartheta(\omega)\beta|\sigma_x|}}{\displaystyle\int_{\mathbb R}e^{-F(s)-\vartheta(\omega)\beta|s|}\diff s}.
\end{equation}
For strongly tempered $\omega$, they provide $\Lambda$-independent lower and upper bounds 
\be\label{pi}
\pi^{\omega}_-(\sigma_x)\le \nu_{\Lambda,\beta}^{\omega}(\sigma_x)\le \pi^{\omega}_+(\sigma_x)
\ee
for all $\sigma_x \in \mathbb{R}$.

For the bounded-spin case, we have further $\omega$-independent bounds replacing $\vartheta(\omega)$ by its upper bound $\vartheta$:
\be\label{pi-bounded}
\pi^\infty_-(\sigma_x)\le \nu_{\Lambda,\beta}^{\omega}(\sigma_x)\le \pi^\infty_+(\sigma_x)
\ee
with
\begin{equation}
\pi^\infty_-(\sigma_x):= \frac{\displaystyle e^{-F(\sigma_x)-\vartheta\beta|\sigma_x|}}{\displaystyle\int_{\mathbb R}e^{-F(s)+\vartheta\beta|s|}\diff s}
    \quad \text{and}  \quad 
\pi^\infty_+(\sigma_x):= \frac{\displaystyle e^{-F(\sigma_x)+\vartheta\beta|\sigma_x|}}{\displaystyle\int_{\mathbb R}e^{-F(s)-\vartheta\beta|s|}\diff s}.
\end{equation}

\subsection{Cluster expansion}

As the first factor in the right-hand side of \eqref{eq:Xuan268_1} is independent of $t$, the identity \eqref{ratio} becomes
\be 
\mu^{\omega}_{\tilde{\Lambda}_k,\beta}\left(\exp\left( \frac{it}{\sqrt{D_k}}\sum_{x\in \tilde{\Lambda}_k}\sigma_x \right)  \right) \;=\; \frac{\Xi^{\omega}_{\tilde{\Lambda}_k,\beta}(t)}{\Xi^{\omega}_{\tilde{\Lambda}_k,\beta}(0)}
\ee 
with
\be
\Xi^{\omega}_{\tilde{\Lambda}_k,\beta}(t)=1+\sum_{n=1}^\infty\sum_{(R_1,\ldots, R_n)\in\mathcal{R}^n}\prod_{1\le i<j\le n}\mathbf{1}_{\underline{R}_i\cap\underline{R}_j=\emptyset}\prod_{i=1}^n\zeta^{(1)}_{t,r_0}(R).
\ee
The cluster expansion of interest here corresponds to the (formal) logarithm of this last polymer expansions.  This takes the form (Theorem \ref{theo:cl1} in Appendix \ref{cluster_expansion_subsec}) 
\begin{equation}\label{eq:Xuan268_2}
\Xi^{\omega}_{\tilde{\Lambda}_k,\beta}(t) = \exp\left(\sum_{n=1}^\infty\frac{1}{n!}\sum_{(R_1,\ldots, R_n)\in\mathcal{R}^n}\omega^T_n(R_1,\ldots, R_n)\prod_{i=1}^n\zeta^{(1)}_{t,r_0}(R)\right)\;=:\; \exp \bigl[U(t)\bigr]\;.
\end{equation}

For $t=0$, the polymers only contain two-body interactions.  To distinguish then, we denote $\mathcal{R}_1$  the set of all elements of the form $\{x\}$ for all $x\in \mathbb{Z}^d$, and $\mathcal{R}_2$ the set of polymers $R\Subset \mathcal{P}_{1,2}$ where all the elements $X\in R$ with $|X|=2$.
With this notation, 
\begin{equation}\label{eq:Xuan268_3} 
\Xi_{\tilde{\Lambda}_k,\beta}^{\omega}(0) =\exp\left(\sum_{n=1}^\infty \frac{1}{n!}\sum_{(R_1,\ldots, R_n)\in \mathcal{R}_{2}^n}\omega^T_n(R_1,\ldots, R_n)\prod_{i=1}^n\zeta^{(1)}_{r_0}(R)\right),
\end{equation}
here, $\zeta^{(1)}_{r_0}(R)=\zeta^{(1)}_{0,r_0}(R)$. 

Combining (\ref{eq:Xuan268_2}) and (\ref{eq:Xuan268_3}),
\begin{align}\label{eq:difference}
\frac{\Xi^{\omega}_{\tilde{\Lambda}_k,\beta}(t)}{\Xi^{\omega}_{\tilde{\Lambda}_k,\beta}(0)}
    &=\exp\Bigg(\sum_{n=1}^{\infty}\frac{1}{n!}\sum_{(R_1,\ldots,R_n) \in \mathcal{R}^n}\omega_n^T(R_1,\ldots,R_n)\prod_{i=1}^n\zeta^{(1)}_{t,r_0}(R)\nonumber\\ 
    &\hspace{2cm}-\sum_{n=1}^{\infty}\frac{1}{n!}\sum_{(R_1,\ldots,R_n) \in \mathcal{R}^n_2}\omega_n^T(R_1,\ldots,R_n)\prod_{i=1}^n\zeta^{(1)}_{r_0}(R)\Bigg).    
\end{align}

The remaining task is to show that for $r_0$ sufficiently large (i) the cluster expansion \eqref{eq:Xuan268_2} converges, and (ii) the real part of the bracketed difference in \eqref{eq:difference} goes to zero sufficiently fast as $k\to\infty$.  

\subsection{``Finite-range" plus perturbation}

If the interaction is finite-range (f.r.), there exists an $r_0$ such that $J_{xy}=0$ if the distance between the sites $x$ and $y$ is strictly larger than $r_0$.  In this case, the sets $\tilde\Lambda_{k}$ become a family of ``isolated" sites, that is, of sites whose spins do not interact with each other.   In this case
\be\label{eq:finite}
\left[\mu^{\omega}_{\tilde{\Lambda}_k,\beta}\left(\exp\Bigl(\frac{it \sigma_x}{\sqrt{D_k}}\Bigr)\right)\right]^{\rm f.r.}\;=\;
\prod_{x\in \tilde{\Lambda}_k}\nu^{\omega}_{\tilde{\Lambda}_k,\beta}\left(\exp\Bigl(\frac{it \sigma_x}{\sqrt{D_k}}\Bigr)\right)\;.
\ee 
While the expression for infinite-range interactions are notoriously more involved, the basic technical steps can be understood at the level of the finite-range product \eqref{eq:finite}.  Heuristically our approach consists on performing a careful estimation of this factorization and subsequently showing that the correction incorporated by long-range contributions does not compete with this leading order.  The smallness of these corrections is determined via cluster-expansion technology. 

These observations can be put into a formal basis by classifying the terms in the expansion \eqref{eq:Xuan268_2} into those consisting purely of polymers in $\mathcal{R}_1$ and the rest.  Due to the connectivity condition, the only clusters of the former type are copies of a single ``monomer" $\{x\}$.  The split takes the form
\be\label{eq:split-u}
U(t)\;=\; U_1(t) + U_2(t)
\ee
with
\begin{align}\label{eq:split-u1}
    U_1(t) &=\sum_{n=1}^\infty\frac{1}{n!}\sum_{x\in \tilde{\Lambda}_k}\omega_n^T(\{x\},\ldots, \{x\})\, \bigl[\zeta^{(1)}_{t,r_0}(\{x\})\bigr]^n\,\\  
    U_2(t)&=\sum_{n=1}^{\infty}\frac{1}{n!}\sum_{\substack{(R_1,\ldots, R_n)\in \mathcal{R}^n\\ \exists 
 i_0\in [n]:\;|\underline{R}_{i_0}|\ge 2}}\omega^T_n(R_1,\ldots, R_n)\,\prod_{i=1}^n\zeta^{(1)}_{t,r_0}(R_i)\;.
\label{eq:split-u2}
\end{align}
From \eqref{eq:polmect12},
\be
\omega_n^T(\{x\},\ldots, \{x\})\;=\; (-1)^{n-1}\,(n-1)!\;.
\ee
Hence
\begin{eqnarray}
U_1(t) &=& \sum_{x\in \tilde{\Lambda}_k} \sum_{n=1}^{\infty}\frac{(-1)^{n-1}}{n} \bigl[\zeta^{(1)}_{t,r_0}(\{x\})\bigr]^n\\
&=& \sum_{x\in \tilde{\Lambda}_k}\,\ln\bigl[1+\zeta^{(1)}_{t,r_0}(\{x\})\bigr]
\end{eqnarray}
Applying the definition
\be
\zeta^{(1)}_{t,r_0}(\{x\})\;=\; \int \left(\exp\left(\frac{it\sigma_x}{\sqrt{D_k}}\right)-1\right)\nu^{\omega}_{\tilde{\Lambda}_k,\beta}(\diff\sigma_x)\;,
\ee
we conclude that [$\nu^{\omega}_{\tilde{\Lambda}_k,\beta}$ is normalized]
\be
\exp\bigl[U_1(t)\bigr]\;=\; \prod_{x\in \tilde{\Lambda}_k}\nu^{\omega}_{\tilde{\Lambda}_k,\beta}\left(\exp\Bigl(\frac{it \sigma_x}{\sqrt{D_k}}\Bigr)\right)
\ee
corresponds to the ``finite-range" part.

Each polymer $R\not\in\mathcal{R}_1$, on the other hand, is formed by a polymer in $\mathcal{R}_2$ plus, possibly, some singletons $\{x\}$ with $x\in \underline{R}$.  The former  is the only part of $R_i$ relevant for the connection condition imposed by $\omega^T$.  Thus, it is convenient to reorganize the sum over polymers into a sum over polymers in $\mathcal{R}_2$  and a sum over families of monomers on top of them.  In this way, 
\be\label{eq:cluster-u2}
U_2(t)\;=\;\sum_{n=1}^{\infty}\frac{1}{n!}\sum_{(R_1,\ldots, R_n)\in \mathcal{R}_2^n} \omega^T_n(R_1,\ldots, R_n)\,\prod_{i=1}^n\Gamma(R_i)(t)
\ee
with
\begin{eqnarray}\label{eq:summed-rho}
\Gamma(R_i)(t)&=&\sum_{S\subset \underline{R}_i}\int_{\mathbb R^{|\underline{R}_i|}}\prod_{\{x,y\}\in R_i}\left(e^{\beta J_{xy}\sigma_x\sigma_y}-1\right)\prod_{x\in S}\left(\exp\left(\frac{it\sigma_x}{\sqrt{D_k}}\right)-1\right)\prod_{x\in\underline{R}_i}\nu^{\omega}_{\tilde{\Lambda}_k,\beta}(\diff\sigma_x)\nonumber\\
&=& \int_{\mathbb R^{|\underline{R}_i|}}\prod_{\{x,y\}\in R_i}\left(e^{\beta J_{xy}\sigma_x\sigma_y}-1\right)\prod_{x\in  \underline{R}_i}\exp\left(\frac{it\sigma_x}{\sqrt{D_k}}\right)\prod_{x\in\underline{R}_i}\nu^{\omega}_{\tilde{\Lambda}_k,\beta}(\diff\sigma_x)
\end{eqnarray}

From the preceding considerations \eqref{eq:difference} takes the promised form ``finite-range-times-perturbation":
\be\label{eq:diff-fin} 
\mu^{\omega}_{\Lambda_k,\beta}(e^{i t \bar{S}_k})\;=\;
\frac{\Xi^{\omega}_{\tilde{\Lambda}_k,\beta}(t)}{\Xi^{\omega}_{\tilde{\Lambda}_k,\beta}(0)}\;=\; 
\left[\prod_{x\in \tilde{\Lambda}_k}\nu^{\omega}_{\tilde{\Lambda}_k,\beta}\left(\exp\Bigl(\frac{it \sigma_x}{\sqrt{D_k}}\Bigr)\right)\right]\, \times\,\exp \bigl[\Delta U(t)\bigr]
\ee
with
\be\label{eq:difference-final}
\Delta U(t)\;=\; \sum_{n=1}^{\infty}\frac{1}{n!}\sum_{(R_1,\ldots, R_n)\in \mathcal{R}_2^n} \omega^T_n(R_1,\ldots, R_n)\,\left[\prod_{i=1}^n\Gamma(R_i)(t) - \prod_{i=1}^n\Gamma(R_i)(0) \right]\;.
\ee

\section{Proof of Theorem \ref{main}. (B) Convergence of the cluster expansion \eqref{eq:cluster-u2}}\label{sec.proof-b}

In the proofs below, the series \eqref{eq:cluster-u2} will be subject to termwise operations, which are strictly rigorous only upon absolute convergence of the series.  This leads to considering the series of positive terms
\be\label{eq:cluster-u2-abs}
\card U_2\;=\;\sum_{n=1}^{\infty}\frac{1}{n!}\sum_{(R_1,\ldots, R_n)\in \mathcal{R}_2^n} \card{\omega^T_n(R_1,\ldots, R_n)}\,\prod_{i=1}^n\card\Gamma(R_i)
\ee
with
\begin{equation}\label{eq:summed-rho-abs}
\card\Gamma(R_i)\;=:\;
\int_{\mathbb R^{|\underline{R}_i|}}\prod_{\{x,y\}\in R_i}\left(e^{\beta \card{J_{xy}\sigma_x\sigma_y}}-1\right)\prod_{x\in\underline{R}_i}\nu^{\omega}_{\tilde{\Lambda}_k,\beta}(\diff\sigma_x)\;.
\end{equation}
Note the fortunate $t$-independence of the bound.

The convergence of \eqref{eq:cluster-u2-abs} implies that $U_2(t)$ is a well defined function that satisfies (i) $\left|U_2(t)\right|\le \card U_2(t)$, and (ii) it is an analytic function of the parameters inside the weights $\Gamma(R)(t)$.  The latter implies, in particular, that successive derivatives of $U_2(t)$ are defined by term-by-term differentiation.   
The following propositions, furthermore, establishes the perturbative character of the contribution $\exp[\Delta U(t)]$.

\begin{proposition}\label{prop:cluster}
For any $\varepsilon>0$:
\begin{itemize}
\item[(i)]  In the bounded-spin case, there exists a dilution $r_{\varepsilon}(\beta)$ such that 
\be
 \card U_2< \bigl|\tilde\Lambda_k\bigr|\,\varepsilon \quad \mbox{for } r_0 \ge r_\varepsilon\;.
\ee
\item[(ii)] In the general unbounded case, there exists a $\beta_\varepsilon(\omega)$ such that  
\be 
\card U_2< \bigl|\tilde\Lambda_k\bigr|\,\varepsilon \quad \mbox{for } \beta \le \beta_\varepsilon\;.
\ee
\end{itemize} 
\end{proposition}

\begin{proof}
We resort to the following well known criterion, reviewed in Appendix \ref{sec:cluster-exp}.   Let $a>0$ such that
\be\label{eq:gk1}
\sup_{x\in \mathbb{Z}^d_{r_0}}\sum_{\substack{R\in \mathcal{R}_2\\ R\ni x }}\card\Gamma(R)\,e^{a|\underline R|} \;\le\; e^a-1.
\ee
then (Corollary \ref{corollary1.1})  the components
\be\label{eq:cluster-u2-R-abs}
\card U_2(R)\;=\;\sum_{n=0}^{\infty}\frac{1}{n!}\sum_{(R_1,\ldots, R_n)\in \mathcal{R}_2^n} \card{\omega^T_n(R,R_1,\ldots, R_n)}\,\prod_{i=1}^n\card\Gamma(R_i)
\ee
satisfy
\be
\card U_2(R)\; \le\; e^{a\card{\underline R}}\;.
\ee
As a result
\be
\card U_2\;\le\; \bigl|\tilde\Lambda_k\bigr| \sup_{x\in \mathbb{Z}^d_{r_0}}\sum_{\substack{R\in \mathcal{R}_2\\ R\ni x }}\card\Gamma(R)\, \card U_2(R)\;\le\; \bigl|\tilde\Lambda_k\bigr| \sup_{x\in \mathbb{Z}^d_{r_0}}\sum_{\substack{R\in \mathcal{R}_2\\ R\ni x }}\card\Gamma(R)\, e^{a\card{\underline R}}
\ee
and hence,
\be
\card{U_2(t)}\;\le\; \bigl|\tilde\Lambda_k\bigr|( e^a-1)\;.
\ee
by \eqref{eq:gk1}.

The following lemma concludes the proof.
\end{proof}

\begin{lemma}\label{lem:cluster2}
For every $a>0$:
\begin{itemize}
\item[(i)]  In the bounded-spin case, there exists a dilution $r_{a}(\beta)$ such that the condition \eqref{eq:gk1} is satisfied for all dilutions $r_0 \ge r_a$.
\item[(ii)] In the general unbounded case, there exists a $\beta_a(\omega)$ such that the condition \eqref{eq:gk1} is satisfied for all dilutions $\beta \le \beta_a$.
\end{itemize} 
\end{lemma}

\begin{proof}
Each polymer $R\in\mathcal R_2$ is uniquely defined by a graph $G(R)$ with vertex set $\underline R$ and edges $\{x,y\}\in R$.  
Thus,
\begin{eqnarray}\label{eq:ta-1}
 T_{x_0,r_0}(a)&:=&\sum_{\substack{R\in \mathcal{R}_2\\ R\ni x_0 }}\card\Gamma(R)\,e^{a|\underline R|}\nonumber\\
&=& \sum_{\substack{C\subset \mathbb{Z}^d_{r_0} \\ C\ni x_0 \\ |C|\ge 2}} e^{a|C|} \sum_{G\in \mathcal G(C)}
\int_{\mathbb R^{|C|}}\prod_{\{x,y\}\in E(G)}\left(e^{\beta \card{J_{xy}\sigma_x\sigma_y}}-1\right)\prod_{x\in C}\nu^{\omega}_{\tilde{\Lambda}_k,\beta}(\diff\sigma_x)
\end{eqnarray}
where $\mathcal G(C)$ denotes all spanning graphs with vertex set $C$.

As reviewed in Appendix \ref{tree_graphs}, an efficient way to account for sums over graphs with vertices $C$ containing a fixed vertex $x_0$ is to use a partition scheme.  This is an injective map that to each graph associates a unique tree $\tau$ in the family $\mathcal T_{x_0}[C]$ of $x_0$-rooted $C$-spanning trees, yielding a partition of the set of graphs into disjoint families of graphs labelled by a tree $\tau$ and containing graphs obtained by adding up to a maximal set $E(\tau)$ of further links.  Then \eqref{eq:ta-1} becomes

\begin{eqnarray}
 \lefteqn{T_{x_0,r_0}(a)}\nonumber\\
 &=&\sum_{\substack{C\subset \mathbb{Z}^d_{r_0} \\ C\ni x_0 \\ |C|\ge 2}} e^{a\card C}
\sum_{\tau\in\mathcal{T}_{x_0}[C]}\int_{\mathbb R^{|C|}}\prod_{\{x,y\}\in E(\tau)}\left(e^{\beta| J_{xy}\sigma_x\sigma_y|}-1\right)\sum_{E(\tau)\subset E \subset E(R(\tau))}
\prod_{\{x,y\}\in E}\left(e^{\beta|J_{xy}\sigma_x\sigma_y|}-1\right)\
\prod_{x\in C}\,\nu^{\omega}_{\tilde{\Lambda}_k,\beta}(\diff\sigma_x)\nonumber\\ 
&=&\sum_{\substack{C\subset \mathbb{Z}^d_{r_0} \\ C\ni x_0 \\ |C|\ge 2}} e^{a\card C}\sum_{\tau\in\mathcal{T}_{x_0}[C]}\int_{\mathbb R^{|C|}}\prod_{\{x,y\}\in E(\tau)}\left(e^{\beta| J_{xy}\sigma_x\sigma_y|}-1\right)\prod_{\{x,y\}\in E(R(\tau))\setminus E(\tau)}e^{\beta|J_{xy}\sigma_x\sigma_y|}
\prod_{x\in C}\nu^{\omega}_{\tilde{\Lambda}_k,\beta}(\diff\sigma_x)\nonumber\\
&=& \sum_{\substack{C\subset \mathbb{Z}^d_{r_0} \\ C\ni x_0 \\ |C|\ge 2}} e^{a\card C}\sum_{\tau\in\mathcal{T}_{x_0}[C]}\int_{\mathbb R^{|C|}}\prod_{\{x,y\}\in E(\tau)}\left(1-e^{-\beta| J_{xy}\sigma_x\sigma_y|}\right)\prod_{\{x,y\}\in E(R(\tau))}e^{\beta|J_{xy}\sigma_x\sigma_y|}
\prod_{x\in C}\nu^{\omega}_{\tilde{\Lambda}_k,\beta}(\diff\sigma_x)
\;.
\end{eqnarray}

To bound this expression we resort to established combinatorial techniques already employed in \cite{EM,PS}, which, in the setting of continuous spins, can also be traced back to \cite{FPS}.  
By Cauchy-Schwarz inequality, and since $\beta |J_{xy}|\le \widehat{\beta}_{r_0}$,
\begin{equation}\label{eq:Xuan218_6}
\beta\sum_{\{x,y\}\in E(R(\tau))}|J_{xy}\sigma_x\sigma_y|\le \sum_{\{x,y\}\subset C}\beta|J_{xy}||\sigma_x\sigma_y|\le \widehat{\beta}_{r_0}\sum_{x\in C}\sigma_x^2.
\end{equation}
Using, furthermore, the elementary bound $1-e^{-x}\le x$, we obtain
\begin{eqnarray}
 T_{x_0,r_0}(a) &\le& 
\sum_{\substack{C\subset \mathbb{Z}^d_{r_0} \\ C\ni x_0 \\ |C|\ge 2}} e^{a\card C} \sum_{\tau\in\mathcal{T}_{x_0}[C]}\int_{\mathbb R^{|C|}}\prod_{\{x,y\}\in E(\tau)}
\beta \bigl| J_{xy}\sigma_x\sigma_y\bigr|
 \prod_{x\in C} e^{2\widehat{\beta}_{r_0}\sigma_x^2}\,\nu^{\omega}_{\tilde{\Lambda}_k,\beta}(\diff\sigma_x)\nonumber\\
 &=& \sum_{n=2}^{\infty}e^{an}\sum_{\substack{C\subset \mathbb{Z}^d_{r_0}\\ C\ni x_0 \\  |C|=n}}\sum_{\tau\in \mathcal{T}_{x_0}[C]}\int_{\mathbb R^{n}}\prod_{\{x,y\}\in E(\tau)}\beta \bigl| J_{xy}\sigma_x\sigma_y \bigr|\prod_{x\in C}e^{2\widehat{\beta}_{r_0}\sigma_x^2}\,\nu^{\omega}_{\tilde{\Lambda}_k,\beta}(\diff\sigma_x)
 \;.
\end{eqnarray}

To continue, let us make explicit reference to the non-$x_0$ vertices in $C$ (note the extra $n!$ due to the sum over \emph{ordered} $n$-tuples).
\begin{eqnarray}
 T_{x_0,r_0}(a) &\le& \sum_{n=1}^{\infty}\frac{e^{a(n+1)}}{n!}\sum_{\substack{x_1,\ldots, x_n\in \mathbb{Z}^d_{r_0}\\ x_i\ne x_j}}\sum_{\tau\in\mathcal{T}_{x_0}[\{x_1,\ldots, x_n\}]}\int_{\mathbb R^{n+1}}\prod_{\{x_i, x_j\}\in E(\tau)}\beta \bigl|J_{x_ix_j}\sigma_{x_i}\sigma_{x_j} \bigr|\prod_{i=0}^n e^{2\widehat{\beta}_{r_0}\sigma_{x_i}^2}
 \,\nu^{\omega}_{\tilde{\Lambda}_k,\beta}(\diff\sigma_{x_i})\nonumber\\
 &=&\sum_{n=1}^{\infty}\frac{1}{n!}\sum_{\substack{x_1,\ldots, x_n\in\mathbb{Z}^d_{r_0}\\ x_i\ne x_j}}\sum_{\tau\in\mathcal{T}_{x_0}[\{x_1,\ldots, x_n\}]}\prod_{\{x_i, x_j\}\in E(\tau)}\beta\bigl|J_{x_ix_j}\bigr|\prod_{i=0}^n\left(e^a\int_{\mathbb R}|\sigma_{x_i}|^{d^\tau_i}e^{2\widehat{\beta}_{r_0}\sigma_{x_i}^2}
 \,\nu^{\omega}_{\tilde{\Lambda}_k,\beta}(\diff\sigma_{x_i})\right)\nonumber\\
 &=:& \sum_{n=1}^{\infty}\frac{1}{n!}\sum_{\substack{x_1,\ldots, x_n\in\mathbb{Z}^d_{r_0}\\ x_i\ne x_j}}\sum_{\tau\in\mathcal{T}_{x_0}[\{x_1,\ldots, x_n\}]}\prod_{\{x_i, x_j\}\in E(\tau)}\beta\bigl|J_{x_ix_j}\bigr|\prod_{i=0}^n \left(e^a\Theta_{\widehat{\beta}_{r_0}}(d^\tau_i)\right)
 \;.
\end{eqnarray}
In this expressions, $d_i^\tau$ is the degree of the vertex $i$ in the tree $\tau$.  At this point we permute the sum over vertices with the sum over trees to obtain
\be\label{eq:permut}
T_{x_0,r_0}(a) \;\le\; \sum_{n=1}^{\infty}\frac{1}{n!} \sum_{\tau\in\mathcal{T}_{x_0}(n)}
\sum_{\substack{x_1,\ldots, x_n\in V_*(\tau)\\ x_i\ne x_j}}\prod_{\{x_i, x_j\}\in E(\tau)}\beta\bigl|J_{x_ix_j}\bigr|\prod_{i=0}^n \left(e^a\Theta_{\widehat{\beta}_{r_0}}(d^\tau_i)\right)
\ee
where $\mathcal{T}_{x_0}(n)$ denotes the family of $x_0$-rooted trees with $n$ additional vertices, and $V_*(\tau)$ the set of non-root vertices of $\tau$.  The sum over vertices can be bounded as follows.
\smallskip

\noindent
{\bf Claim:} For any $\tau\in \mathcal{T}_{x_0}(n)$,
\be 
f(\tau):=\sum_{\substack{x_1,\ldots, x_n\in V_*(\tau)\\ x_i\ne x_j}} \prod_{\{x_i,x_j\} \in E(\tau)} \beta|J_{x_ix_j}| 
\;\le\; \widehat\beta_{r_0}^{|E(\tau)|}\;=\;\widehat{\beta}_{r_0}^n.
\ee
This is proven by induction on $k$ = maximum depth of the tree. For $k=1$, we have
\be
\sum_{\substack{x_1,\ldots, x_n\in V_*(\tau)\\ x_i\ne x_j}}  \prod_{i=1}^n \beta|J_{x_0x_i}| \le \prod_{i=1}^n \biggl( \sum_{x_i \in \mathbb{Z}^d_{r_0}} \beta |J_{x_0x_i}| \biggr) \;\le\;  \widehat{\beta}^n_{r_0}.
\ee
Assume that the inequality holds for all finite trees with maximum depth $k-1$ and consider a tree $\tau$ with depth $k$.  The tree is defined by the vertices $x_{i_1},\ldots,x_{i_\ell}$ that are connected to the root $x_0$, and the forest formed by  trees $\tau_{i_j}$ with root $x_{i_j}$, $1\le j\le \ell$, which is created upon removal of the root $x_0$. Note that,
\be
\card{E(\tau)}\;=\; \ell + \sum_{j=1}^\ell \card{E(\tau_{i_j})}\;.
\ee
Then
\begin{align*}
\sum_{\substack{x_1,\ldots, x_n\in\mathbb{Z}^d_{r_0}\\ x_i\ne x_j, 0\le i,j\le n}} &\prod_{\{x_i,x_j\} \in E(\tau)} \beta|J_{x_ix_j}| 
=  \sum_{x_{i_1},\ldots,x_{i_\ell}} \prod_{j=1}^{\ell} \beta |J_{x_0x_{i_j}}| \prod_{j=1}^{\ell}f(\tau_{i_j})\\
&\le \sum_{x_{i_1},\ldots,x_{i_\ell}} \prod_{j=1}^{\ell} \beta |J_{x_0x_{i_j}}| \prod_{j=1}^{\ell} \widehat{\beta}_{r_0}^{|E(\tau_{i_j})|}
\;\le\; \widehat{\beta}_{r_0}^{\ell}\,\prod_{j=1}^{\ell} \widehat{\beta}_{r_0}^{|E(\tau_{i_j})|}
\;=\; \widehat{\beta}^n_{r_0}\;.
\end{align*}
The claim is proven.
\smallskip

This claim and \eqref{eq:permut} lead to
\be 
 T_{x_0,r_0}(a) \;\le\; \sum_{n=1}^\infty\frac{\widehat{\beta}_{r_0}^{n}}{n!}\sum_{\tau\in\mathcal{T}_{x_0}(n)}\prod_{i=0}^n\left(e^a\Theta_{\widehat{\beta}_{r_0}}(d^\tau_i)\right)\;.
\ee
As the last summand is only a function of the degrees of the tree, we can pass to a sum over degrees weighted by the number (Cayley formula)
\[
\sum_{\tau\in\mathcal{T}_{d_0,\ldots,d_n}}1=\frac{(n-1)!}{\prod_{i=0}^n(d_i-1)!},
\]
where $\mathcal{T}_{d_0,\ldots,d_n}$ is the set of labeled trees with $n+1$ vertices and each vertex $i$ has degree $d_i$.
Then, 
\begin{eqnarray}\label{eq:tfin}
    T_{x_0,r_0}(a) 
    &\le& \sum_{n=1}^{\infty}\widehat{\beta}_{r_0}^n\,e^{a(n+1)}\sum_{\substack{d_0+\ldots+d_{n}=2n\\ d_i\ge 1}}\prod_{i=0}^n\frac{\Theta_{\widehat{\beta}_{r_0}}(d_i)}{(d_i-1)!}\nonumber\\ 
    &\le&\sum_{n=1}^{\infty} \widehat{\beta}_{r_0}^n
    \left(e^a\sum_{s\ge 0}\frac{\Theta_{\widehat{\beta}_{r_0}}(s+1)}{s!}\right)^{n+1}\nonumber\\
    &=:& \sum_{n=1}^{\infty} \widehat{\beta}_{r_0}^n \bigl(e^a\,G_{\widehat{\beta}_{r_0}}\bigr)^{n+1}
\end{eqnarray}    
Let us observe that
\begin{eqnarray*}
G_{\widehat{\beta}_{r_0}}\;=\; \sum_{s\ge 0}\frac{1}{s!}
\int_{\mathbb R}|\sigma_{x}|^{s}e^{2\widehat{\beta}_{r_0}\sigma_{x}^2}
 \,\nu^{\omega}_{\tilde{\Lambda}_k,\beta}(\diff\sigma_{x})\;\le\; 
 \int_{\mathbb R}\exp\bigl(|\sigma_{x}|+2\widehat{\beta}_{r_0}\sigma_{x}^2\bigr)\,\pi^{\omega}_+(\diff\sigma_{x})\;.
\end{eqnarray*}
This integral exists as long as $2\widehat{\beta}_{r_0}$ is strictly less than the superstability constant $A$ in Definition \ref{superstability}.
In fact, imposing the dilution condition
\be\label{eq:dil-1}
\widehat{\beta}_{r_0}\;\le\; \frac{A}{4}
\ee
(recall that $r_0=1$ in the unbounded-spin case) we obtain the dilution-independent bound
\be
G_{\widehat{\beta}_{r_0}}\;\le\; \int_{\mathbb R}\exp\Bigl(|\sigma_{x}|+\frac{A}{4}\sigma_{x}^2\Bigr)\,\pi^{\omega}_+(\diff\sigma_{x})
\;=:\; G^\omega\;.
\ee
For the bounded spin case, the bound becomes $\omega$-independent by replacing $\omega$ by $\infty$ in this expression.
Inserting this bound in inequality \eqref{eq:tfin} and asking the additional conditions
\be\label{eq:dil-2}
\widehat{\beta}_{r_0}e^aG^\infty\le\frac12\quad\mbox{and}\quad 
\widehat{\beta}_{1}e^aG^\omega\le\frac12
\ee
respectively for bounded and unbounded spins,
leads to the final bound
\be
 T_{x_0,r_0}(a) \;\le\; 2(e^a G^{\omega})^2\widehat{\beta}_{r_0}\quad\mbox{and}\quad T_{x_0,r_0}(a) \;\le\; 2(e^a G^{\infty})^2\;.
\ee
As this bound is independent of $x_0$, condition  \eqref{eq:gk1} is satisfied if 
\be\label{eq:dil-3}
\widehat{\beta}_{r_0}\;\le\; \frac{e^a-1}{2(e^a G^{\omega})^2}\quad\mbox{and}\quad \widehat{\beta}_{r_0}\;\le\; \frac{e^a-1}{2(e^a G^{\infty})^2} \;. 
\ee
The lemma is proven, with $r_a$ and $\beta_a$ determined by the three conditions \eqref{eq:dil-1}, \eqref{eq:dil-2} and \eqref{eq:dil-3}.
\end{proof}

\section{Proof of Theorem \ref{main}. (C) Bound for the soft-frequency integral}

The bound is a result of the following proposition.

\begin{proposition}\label{prop.xuan1} 
For given $\beta$ and strongly tempered boundary condition $\omega$, there exist $\delta>0$ and $C^{(1)}(\beta,\omega)>0$ such that, for every $|t| \leq \delta \sqrt{D_{k}}$,
\be\label{eq:cbound-1}
\left|\mu^{\omega}_{\Lambda_k,\beta}(e^{i t \bar{S}_k})\right| \leq \exp \left(-t^{2}C^{(1)}\frac{|\tilde{\Lambda}_{k}|}{D_{k}}\right)
\ee  
for every $k>1$.  For the bounded-spin case the constant $C^{(1)}$ can be chosen uniformly in $\omega$.
\end{proposition}

Before proceeding to the proof, let us see why this bound yields the desired integral bound.  Indeed, as $\tilde \Lambda_k/D_k| \to r_0^d/L$ due to conditions (i)-(ii) of the local CLT (Definition \ref{def:lclt}), we obtain that for $k$ large enough
\be\label{eq:int-soft}
\int_{\bigl\{B\le \card t\le \delta \sqrt{D_k}\bigr\}}  \left|\mu^{\omega}_{\Lambda_k,\beta}(e^{i t \bar{S}_k})\right| \diff t
\;\le\; \int_B^\infty e^{-t^2 C^{(1)}(r_0^d/L-\varepsilon)} \diff t 
\ee
and the right-hand side can be made arbitrarily small by taking $B$ sufficiently large.  For bounded spins, this bound is uniform in $\omega$, and hence the supremum in \eqref{eq:bound-rfr} applies.

The proof of the proposition is divided in two parts, respectively addressing the finite-range and the perturbation parts in \eqref{eq:diff-fin}:
\begin{description}
\item[First part (Lemma \ref{lem:prodnu}):] For $\delta$ sufficiently small, there exists $c_1>0$ such that
\be
\card{\prod_{x\in \tilde{\Lambda}_k}\nu^{\omega}_{\tilde{\Lambda}_k,\beta}\left(\exp\Bigl(\frac{it \sigma_x}{\sqrt{D_k}}\Bigr)\right)}
\;\le\; \exp \left(-t^{2} c_1\frac{|\tilde{\Lambda}_{k}|}{D_{k}}\right)
\ee
for all $t$ with $\card t\le \delta \sqrt{D_k}$.
\item[Second part (Lemma \ref{lem:deltau}):] Given $\varepsilon>0$, for $r_0$ sufficiently large,
\be
\card{e^{\Delta U(t)}}\;\le\; \exp \left(t^{2} \varepsilon\frac{|\tilde{\Lambda}_{k}|}{D_{k}}\right)
\ee 
for all $t>0$.
\end{description}

The combination of the lemmas proves \eqref{eq:cbound-1} with $C^{(1)}=c_1-\varepsilon$.

Following an already established practice \cite{CDT,EM,FP}, in the proofs below we will resort to the second-order Taylor identity
\be\label{eq:taylor2}
F(t)\;=\; F(0) +t F'(0) +\frac{t^2}{2}F''(\theta)
\ee
for some $\theta(F)\in[0,t]$.

\begin{lemma}\label{lem:prodnu}

There exists $\delta_0(\beta,\omega)$ such that if $\delta\le \delta_0$, and there exists $c_1(\beta,\omega)>0$ with
\be\label{eq:prodnu}
\card{\nu^{\omega}_{\tilde{\Lambda}_k,\beta}\left(\exp\Bigl(\frac{it \sigma_x}{\sqrt{D_k}}\Bigr)\right)}
\;\le\; e^{-t^{2} c_1/D_{k}}
\ee
for all $t$ with $\card t\le \delta \sqrt{D_k}$. For bounded spins, both $\delta_0$ and $c_1$ can be chosen uniformly in $\omega$.  
\end{lemma}

\begin{proof}
The use of logarithms turns out to be convenient.  Write
\be
\nu^{\omega}_{\tilde \Lambda_k,\beta}\left(\exp\Bigl(\frac{it \sigma_x}{\sqrt{D_k}}\Bigr)\right)\;=\; \exp\left[\ln\left( \nu^{\omega}_{\tilde \Lambda_k,\beta}\left(\exp\Bigl(\frac{it \sigma_x}{\sqrt{D_k}}\Bigr)\right)\right)\right] 
\ee
and apply to the exponent the second order Taylor identity \eqref{eq:taylor2}.  The result is
\be\label{eq:m2m1}
\nu^{\omega}_{\tilde \Lambda_k,\beta}\left(\exp\Bigl(\frac{it \sigma_x}{\sqrt{D_k}}\Bigr)\right)\;=\; \exp\Bigl[ i\frac{t}{\sqrt{D_k}}M_1^{(k)}(0) - \frac{t^2}{2 D_k}\bigl(M_2^{(k)}(\theta)-[M_1^{(k)}(\theta)]^2\bigr)\Bigr]
\ee
with
\be
M_\ell^{(k)}(t)\;=\; \frac{ \nu^{\omega}_{\tilde \Lambda_k,\beta}\Bigl[\sigma_x^\ell\,\exp\Bigl[(i\frac{t}{\sqrt{D_k}}\sigma_x\Bigr)\Bigr]}{\nu^{\omega}_{\tilde \Lambda_k,\beta}\Bigl[\exp\Bigl[(i\frac{t}{\sqrt{D_k}}\sigma_x\Bigr)\Bigr]}.
\ee
Hence,
\be\label{eq:inter}
\card{\nu^{\omega}_{\tilde{\Lambda}_k,\beta}\left(\exp\Bigl(\frac{it \sigma_x}{\sqrt{D_k}}\Bigr)\right)}\;=\;
\exp\Bigl[  - \frac{t^2}{2 D_k}{\mathrm Re}\bigl(M_2^{(k)}(\theta)-[M_1^{(k)}(\theta)]^2\bigr)\Bigr]\;.
\ee

In detail,
\be\label{eq:abab}
M_2^{(k)}(\theta)-[M_1^{(k)}(\theta)]^2\;=\;\frac{C+iD}{A+iB}-\biggl(\frac{E+iF}{A+iB}\biggr)^2
\ee
with
\begin{equation}
    A=\int_{\mathbb R}\cos\left(\frac{\theta}{\sqrt{D_k}}\sigma_x\right)\nu^{\omega}_{\tilde{\Lambda}_k,\beta}(\diff \sigma_x)\;,\quad 
    B= \int_{\mathbb R}\sin\left(\frac{\theta}{\sqrt{D_k}}\sigma_x\right)\nu^{\omega}_{\tilde{\Lambda}_k,\beta}(\diff\sigma_x)\;,
\end{equation}
\begin{equation}
    C=\int_{\mathbb R}\cos\left(\frac{\theta}{\sqrt{D_k}}\sigma_x\right)\sigma_x^2\,\nu^{\omega}_{\tilde{\Lambda}_k,\beta}(d\sigma_x)\;,\quad 
    D= \int_{\mathbb R}\sin\left(\frac{\theta}{\sqrt{D_k}}\sigma_x\right)\sigma_x^2\,\nu^{\omega}_{\tilde{\Lambda}_k,\beta}(\diff\sigma_x)\;, 
\end{equation}
\begin{equation}
    E=\int_{\mathbb R}\cos\left(\frac{\theta}{\sqrt{D_k}}\sigma_x\right)\sigma_x\,\nu^{\omega}_{\tilde{\Lambda}_k,\beta}(\diff\sigma_x)\;,\quad 
    F= \int_{\mathbb R}\sin\left(\frac{\theta}{\sqrt{D_k}}\sigma_x\right)\sigma_x\,\nu^{\omega}_{\tilde{\Lambda}_k,\beta}(\diff\sigma_x)\;.
\end{equation}
These expressions can be bound using the relations
\be
\cos x=1-\frac{x^2}{2}\cos \gamma\quad,\quad \sin x= x \cos\gamma\quad,\quad \card \gamma \le \card x\;.
\ee
The leftmost equality  is due to the second-order Taylor identity \eqref{eq:taylor2}, while the rightmost one is due to the first-order Taylor identity (a.k.a. mean-value theorem).  They imply
\begin{align*}
\card{A-1}&\le\delta^2M_2^+\\
\card B &\le \delta M_1^+\\
\card{C-M_2^{(k)}(0)} &\le \delta^2 M_4^+\\
\card D &\le \delta M_3^+\\
\card{E-M_1^{(k)}(0)} &\le \delta^2 M_3^+\\
\card F &\le \delta M_2^+
\end{align*}
with
\be
M_\ell^+\;=\;\int_{\mathbb R}\card{\sigma_x}^\ell\,\pi_+^\omega(\diff\sigma_x)\;.
\ee
Thus,
\begin{eqnarray}
{\mathrm Re}\bigl[M_2^{(k)}(\theta)-[M_1^{(k)}(\theta)]^2\bigr] 
&=& {\mathrm Re}\biggl[\frac{M_2^{(k)}(0)+O(\delta)}{1+O(\delta)}-\biggl(\frac{M_1^{(k)}(0)+O(\delta)}{1+O(\delta)}\biggr)^2\biggr]\nonumber\\[5pt]
&=& M_2^{(k)}(0)-[M_1^{(k)}(0)]^2+O(\delta)\;.
\end{eqnarray}
To obtain a $k$-independent bound we proceed as follows
\begin{eqnarray}\label{eq:picov}
M_2^{(k)}(0)-[M_1^{(k)}(0)]^2&=& \int_{\mathbb R}\int_{\mathbb R}\bigl[ \card{\sigma_k}^2-\card{\sigma_x}\card{\sigma_y}\bigr]\,\nu^{\omega}_{\tilde{\Lambda}_k,\beta}(\diff\sigma_x)\,\nu^{\omega}_{\tilde{\Lambda}_k,\beta}(\diff\sigma_y)\nonumber\\
&=& \frac12 \int_{\mathbb R}\int_{\mathbb R}\bigl[ \card{\sigma_x}-\card{\sigma_y}\bigr]^2\,\nu^{\omega}_{\tilde{\Lambda}_k,\beta}(\diff\sigma_x)\,\nu^{\omega}_{\tilde{\Lambda}_k,\beta}(\diff\sigma_y)\nonumber\\
&\ge&  \frac12 \int_{\mathbb R}\int_{\mathbb R}\bigl[ \card{\sigma_x}-\card{\sigma_y}\bigr]^2\,\pi^{\omega}_-(\diff\sigma_x)\,\pi^{\omega}_-(\diff\sigma_y)\\[5pt]
&=& M_2^- - [M_1^-]^2\nonumber
\end{eqnarray}
with
\be
M_\ell^-\;=\; \int_{\mathbb R} \card{\sigma_x}^\ell\,\pi_-^\omega(\diff \sigma_x)\;.
\ee
This proves \eqref{eq:prodnu} with $c_1$ slightly less than $\bigl(M_2^--[M_1^-]^2\bigr)/2$, which is a positive number as shown by \eqref{eq:picov}.  

For bounded spins, $\omega$-uniform bounds are obtained by replacing in all the previous bounds $\pi^\omega_\pm$ by $\pi^\infty_\pm$.
\end{proof}

\begin{lemma}\label{lem:deltau}
For every $\varepsilon>0$ there exists 
\begin{itemize}
\item[(i)] for bounded spins, a dilution $r_{\varepsilon}(\beta)$ and a constant $c_1(\beta,\omega)>0$, and
\item[(ii)] for unbounded spins ($r_0=1$) an inverse temperature $\beta_\varepsilon(\omega)$ and a constant $c_1(\beta,\omega)>0$,
\end{itemize}
such that
\be
\card{e^{\Delta U(t)}}\;\le\; \exp \left(t^{2} \varepsilon\frac{|\tilde{\Lambda}_{k}|}{D_{k}}\right)
\ee 
for all $t>0$ and, respectively, (i) all $r_0\ge r_\varepsilon$, and (ii) all $\beta\le \beta_\varepsilon$.
\end{lemma}

\begin{proof}

We apply the second-order Taylor expansion to $\Delta U(t)$.  By the convergence result proven in Proposition \ref{prop:cluster}, the derivatives are the series defined by termwise differentiation.  It is straightforward to see that $[\Delta U]'(t)$ is purely imaginary, hence
\be
\card{\exp\bigl[\Delta U(t)\bigr]}\;=\; \exp\Big[-\frac{t^2}{2} {\mathrm Re} \, U_2''(\theta)\Bigr]
\ee
for some $\theta\in[0,t]$, and
\be\label{eq:u2derivative}
 U_2''(t)\;=\; \sum_{n=1}^{\infty}\frac{1}{n!}\sum_{(R_1,\ldots, R_n)\in \mathcal{R}_2^n} \omega^T_n(R_1,\ldots, R_n)\,\frac{d^2}{dt^2}\left[\prod_{i=1}^n\Gamma(R_i)(t)\right]\;.
\ee
We have 
\be
\frac{d^2}{dt^2}\left[\prod_{i=1}^n\Gamma(R_i)\right] \;=\; \sum_{i=1}^n \Gamma''(R_i) \prod_{\substack{1\le j\le n\\ j\neq i}}\Gamma(R_j)
+ \sum_{\substack{1\le i, j\le n\\ i\neq j}}  \Gamma'(R_i) \Gamma'(R_j) \prod_{\substack{1\le k\le n\\ k\neq i,j}}\Gamma(R_k)
\ee
with
\begin{eqnarray}
\Gamma'(R)(t) &=& \frac{1}{\sqrt{D_k}}\int_{\mathbb R^{|\underline{R}|}}\Bigl[\prod_{\{x,y\}\in R}\left(e^{\beta J_{xy}\sigma_x\sigma_y}-1\right)\Bigr]\,\sigma_{\underline R} \,\exp\left(\frac{it\sigma_{\underline R}}{\sqrt{D_k}}\right)\nu^{\omega}_{\tilde{\Lambda}_k,\beta}(\diff\sigma^{\underline R})\\
\Gamma''(R)(t) &=& \frac{1}{D_k} \int_{\mathbb R^{|\underline{R}|}} \Bigl[\prod_{\{x,y\}\in R}\left(e^{\beta J_{xy}\sigma_x\sigma_y}-1\right)\Bigr] \, (\sigma_{\underline R})^2\,\exp\left(\frac{it\sigma_{\underline R}}{\sqrt{D_k}}\right)\nu^{\omega}_{\tilde{\Lambda}_k,\beta}(\diff\sigma^{\underline R})\;.
\end{eqnarray}

For compactness, we are denoting
\be
\sigma_C:=\sum_{x\in C} \sigma_x\quad,\quad 
\nu^{\omega}_{\tilde{\Lambda}_k,\beta}(\diff\sigma^C):= \prod_{x\in C}\nu^{\omega}_{\tilde{\Lambda}_k,\beta}(\diff\sigma_x)\;.
\ee

Proceeding as in the proof of Lemma \ref{lem:cluster2}, we bound
\begin{align}
\card{\Gamma'(R)(t)}&\le\; \card{\Gamma'}(R) \;:=\; \frac{1}{\sqrt{D_k}}\int_{\mathbb R^{|\underline{R}|}}\Bigl[\prod_{\{x,y\}\in R}\left(e^{\beta \card{J_{xy}\sigma_x\sigma_y}}-1\right)\Bigr]\,\card{\sigma_{\underline R}} \,\nu^{\omega}_{\tilde{\Lambda}_k,\beta}(\diff\sigma^{\underline R})\\
\card{\Gamma''(R)(t)}&\le\; \card{\Gamma''}(R)\;:=\; \frac{1}{D_k} \int_{\mathbb R^{|\underline{R}|}} \Bigl[\prod_{\{x,y\}\in R}\left(e^{\beta \card{J_{xy}\sigma_x\sigma_y}}-1\right)\Bigr] \, \card{\sigma_{\underline R}}^2\,\exp\left(\frac{it\sigma_{\underline R}}{\sqrt{D_k}}\right)\nu^{\omega}_{\tilde{\Lambda}_k,\beta}(\diff\sigma^{\underline R})\;.
\end{align}
In turns, these bounds imply that
\begin{eqnarray}
\lefteqn{\card{\frac{d^2}{dt^2}\left[\prod_{i=1}^n\Gamma(R_i)\right]}}\nonumber\\
&=&\frac{1}{D_k}\int_{\mathbb R^{\sum_i|\underline{R_i}|}} \biggl[\prod_{i=1}^n \prod_{\{x,y\}\in R_i}\left(e^{\beta \card{J_{xy}\sigma_x\sigma_y}}-1\right)\biggr]\biggl[\sum_{i=1}^n \card{\sigma_{\underline R_i}}^2 +  \sum_{\substack{1\le i, j\le n\\ i\neq j}}  \card{\sigma_{\underline R_i}} \bigl|\sigma_{\underline R_j}\bigr| \biggr] \prod_{i=1}^n \nu^{\omega}_{\tilde{\Lambda}_k,\beta}(\diff\sigma^{\underline R_i})\nonumber\\
&=&\frac{1}{D_k}\int_{\mathbb R^{\sum_i|\underline{R_i}|}} \biggl[\prod_{i=1}^n \prod_{\{x,y\}\in R_i}\left(e^{\beta \card{J_{xy}\sigma_x\sigma_y}}-1\right)\biggr]\biggl[\sum_{i=1}^n \card{\sigma_{\underline R_i}}\biggr]^2 \prod_{i=1}^n \nu^{\omega}_{\tilde{\Lambda}_k,\beta}(\diff\sigma^{\underline R_i})\nonumber\\
&\le&\frac{1}{D_k}\int_{\mathbb R^{\sum_i|\underline{R_i}|}} \biggl[\prod_{i=1}^n \prod_{\{x,y\}\in R_i}\left(e^{\beta \card{J_{xy}\sigma_x\sigma_y}}-1\right)\biggr]\exp\biggl(2\sum_{i=1}^n \card{\sigma_{\underline R_i}}\biggr) \prod_{i=1}^n \nu^{\omega}_{\tilde{\Lambda}_k,\beta}(\diff\sigma^{\underline R_i})\;.
\end{eqnarray}
Also, for the single-polymer cluster,
\be
\card{\Gamma''}(R)\;\le\; \frac{1}{D_k} \int_{\mathbb R^{|\underline{R}|}} \Bigl[\prod_{\{x,y\}\in R}\left(e^{\beta \card{J_{xy}\sigma_x\sigma_y}}-1\right)\Bigr] \, e^{2\sigma_{\underline R}}\,\nu^{\omega}_{\tilde{\Lambda}_k,\beta}(\diff\sigma^{\underline R})\;.
\ee

Putting all of the above together, we conclude that
\be
\card{U''_2(t)}\;\le\; \card{U''}_2 \; :=\;
\sum_{n=1}^{\infty}\frac{1}{n!}\sum_{(R_1,\ldots, R_n)\in \mathcal{R}_2^n} \card{\omega^T_n(R_1,\ldots, R_n)}\,\left[\prod_{i=1}^n \bigl|\widetilde \Gamma\bigr|(R_i)\right]
\ee
with
\be
\bigl|\widetilde \Gamma\bigr|(R)\;=\; 
\int_{\mathbb R^{|\underline{R}|}}\prod_{\{x,y\}\in R}\left(e^{\beta \card{J_{xy}\sigma_x\sigma_y}}-1\right)\prod_{x\in\underline{R}} e^{2\sigma_x}\,\nu^{\omega}_{\tilde{\Lambda}_k,\beta}(\diff\sigma_x)\;.
\ee
\smallskip

At this point we can repeat verbatim the proof of Proposition \ref{prop:cluster} changing $\card \Gamma(R) \to \bigl|\widetilde\Gamma\bigr|(R)$ and
\begin{eqnarray*}
\Theta(s) & \to & \widetilde\Theta(s) := \int_{\mathbb R}|\sigma_{x_i}|^{s}e^{2\widehat{\beta}_{r_0}\sigma_{x_i}^2+2\card{\sigma_x}}
 \,\nu^{\omega}_{\tilde{\Lambda}_k,\beta}(\diff\sigma_{x_i})\\ 
G & \to & \widetilde G:= \int_{\mathbb R}\exp\Bigl(3|\sigma_{x}|+\frac{A}{4}\sigma_{x}^2\Bigr)\,\pi^{\omega}_+(\diff\sigma_{x})\;,
\end{eqnarray*}
concluding the proof of (i).  The proof of (ii) is analogous, replacing in all bounds $\pi^\omega_\pm$ by $\pi^\infty_\pm$.
\end{proof}

\section{Proof of Theorem \ref{main}. (D) Bound for the medium-frequency integral}

The bound relies on the following lemma.
\begin{lemma}\label{lem:crucial}
For each $0<\delta<T$ there exists $\eta(\delta,T)>0$ such that 
\be\label{eq:crucial}
    \left|\nu^{\omega}_{\tilde{\Lambda}_k,\beta}\left(\exp\left(\frac{it\sigma_x}{\sqrt{D_k}}\right)\right)\right|\le e^{-\eta}.
\ee
for all $t$ with $\delta\sqrt{D_k}\le |t|\le T\sqrt{D_k}$.  The value of $\eta$ depends, in general, both of $\beta$ and $\omega$, but in the finite-spin case can be chosen independent of $\omega$.
\end{lemma}
Indeed, this condition is enough because, by \eqref{eq:diff-fin} and Proposition
\ref{prop:cluster}
\be
\left|\mu^{\omega}_{\Lambda_k,\beta}(e^{it\bar{S}_k})\right| \;\le\; \exp\bigl[-|\tilde \Lambda_k| (\eta-2\epsilon)\bigr]
\ee
with $\eta-2\epsilon>0$ if $r_0$ is large enough (finite spins) or $\beta$ small enough (unbounded spins).  Therefore, as $D_k=O(|\Lambda_k|)$ due to conditions (i)-(ii) of the local CLT (Definition \ref{def:lclt}),
\be\label{eq:soft.limit}
\lim_{k\to \infty}  \int_{\delta\sqrt{D_k}}^{T\sqrt{D_k}} \left|\mu^{\omega}_{\Lambda_k,\beta}(e^{it\bar{S}_k})\right| \diff t
\;\le\; \lim_{k\to \infty} (T-\delta) \sqrt{D_k} \, e^{-|\tilde \Lambda_k| (\eta-2\epsilon)}\;=\;0\;.
\ee
In the finite-spin case, this convergence is uniform in $\omega$, a fact that justifies the use of dilution.

\begin{proof}
We resort to duplicated variables
    \begin{align*}
        &\left|\nu^{\omega}_{\tilde{\Lambda}_k,\beta}\left(\exp\left(\frac{it\sigma_x}{\sqrt{D_k}}\right)\right)\right|^2 
        =\left(\int_{\mathbb R}\cos\left(\frac{t\sigma_x}{\sqrt{D_k}}\right)\nu_{\tilde{\Lambda}_k,\beta}^{\omega}(\diff\sigma_x)\right)^2+\left(\int_{\mathbb R}\sin\left(\frac{t\sigma_x}{\sqrt{D_k}}\right)\nu_{\tilde{\Lambda}_k,\beta}^{\omega}(\diff\sigma_x)\right)^2\\ 
        &=\int_{\mathbb R}\int_{\mathbb R}\cos\left(\frac{t\sigma_x}{\sqrt{D_k}}\right)\cos\left(\frac{t\sigma_y}{\sqrt{D_k}}\right)\nu_{\tilde{\Lambda}_k,\beta}^{\omega}(\diff\sigma_x)\nu_{\tilde{\Lambda}_k,\beta}^{\omega}(\diff\sigma_y) 
        +\int_{\mathbb R}\int_{\mathbb R}\sin\left(\frac{t\sigma_x}{\sqrt{D_k}}\right)\sin\left(\frac{t\sigma_y}{\sqrt{D_k}}\right)\nu_{\tilde{\Lambda}_k,\beta}^{\omega}(\diff\sigma_x)\nu_{\tilde{\Lambda}_k,\beta}^{\omega}(\diff\sigma_y)\\ 
        &=1-2\int_{\mathbb R}\int_{\mathbb R}\sin^2\left(\frac{t}{2\sqrt{D_k}}(\sigma_x-\sigma_y)\right)\nu_{\tilde{\Lambda}_k,\beta}^{\omega}(\diff\sigma_x)\nu_{\tilde{\Lambda}_k,\beta}^{\omega}(\diff\sigma_y)\\
        &\le1-2\int_{\mathbb R}\int_{\mathbb R}\sin^2\left(\frac{t}{2\sqrt{D_k}}(\sigma_x-\sigma_y)\right)\pi_-^{\omega}(\diff\sigma_x)\pi_-^{\omega}(\diff\sigma_y)
    \end{align*}
To obtain an upper bound we restrict the last integral to the set
\[
\mathrm{J}_T:=\left\{(\sigma_x,\sigma_y)\in\mathbb R\times\mathbb R\bigg|\frac{2}{T}\le |\sigma_x-\sigma_y|\le  \frac{\pi}{T}\right\}\;.
\]
Note that if $\delta\sqrt{D_k}\le|t|\le T\sqrt{D_k}$, the argument of the integrand satisfies 
\be
\card{\frac{t}{2\sqrt{D_k}}(\sigma_x-\sigma_y)}\in \Big[\frac{\delta}{T}\,, \frac{\pi}{2}\Bigr]\subset
\Bigl[0, \frac{\pi}{2}\Bigr]\;.
\ee
Within this interval, the function $\sin^2$ takes its minimum at the left endpoint.  Thus,
\be
\left|\nu^{\omega}_{\tilde{\Lambda}_k,\beta}\left(\exp\left(\frac{it\sigma_x}{\sqrt{D_k}}\right)\right)\right|^2\;\le\;
1-2\, \sin^2\left(\frac{\delta}{T}\right) \iint_{\mathrm{J}_T}\pi_-^{\omega}(\diff\sigma_x)\pi_-^{\omega}(\diff\sigma_y)\;=:\; e^{-2\eta}\;.
\ee
For bounded spins, an $\omega$-independent $\eta$ is found by replacing $\pi^\omega_\pm$ by $\pi^\infty_\pm$.
\end{proof}

\section{Proof of Theorem \ref{main}. (E) Bound for the hard-frequency integrals}\label{sec:hard}

Once again, the proof relies non a suitable bound for the ``finite-range'' part.  A simple integration by parts yield the following key estimation.

\begin{lemma}\label{lem:hard}

There exists a constant $\gamma(\beta,\omega)$ such that
\be\label{eq:hard-key}
\left|\nu^{\omega}_{\tilde{\Lambda}_k,\beta}\left(\exp\left(\frac{it\sigma_x}{\sqrt{D_k}}\right)\right)\right|\;\le\;
\frac{\gamma\sqrt{D_k}}{\card t}\;.
\ee
For bounded spins, the constant $\gamma$ can be chosen depending only on $\beta$.
\end{lemma}

The desired bound follows immediately because, then, by \eqref{eq:diff-fin} and Proposition
\ref{prop:cluster},
\be
\left|\mu^{\omega}_{\Lambda_k,\beta}(e^{it\bar{S}_k})\right| \;\le\; \biggl(\frac{\gamma\sqrt{D_k}e^{2\epsilon}}{\card t}\biggr)^{\card{\tilde\Lambda_k}}
\ee
and
\be
\int_{T\sqrt{D_k}}^\infty \left|\mu^{\omega}_{\Lambda_k,\beta}(e^{it\bar{S}_k})\right| \diff t\;\le\; \frac{\gamma\sqrt{D_k}e^{2\epsilon}}{ \bigl|\tilde\Lambda_k\bigr|-1} 
\Bigl(\frac{\gamma\,e^{2\epsilon}}{T}\Bigr)^{\card{\tilde\Lambda_k}-1}
\ee
which converges to zero as $k\to\infty$, as long as $T>\gamma e^{2\epsilon}$.  In the finite-spin case, the $\omega$-uniformity of $\gamma$ justifies the use of dilution.

\begin{proof}
Integrating by parts the numerator and resorting to the superstability condition \eqref{superstability},

\begin{eqnarray}
\lefteqn{\nu_{\tilde{\Lambda}_k,\beta}^{\omega}\left(\exp\left(\frac{it\sigma_x}{\sqrt{D_k}}\right)\right)}\nonumber\\[8pt]
 &=& \frac{\displaystyle\int_{\mathbb R}
\exp\biggl(-F(\sigma_x)+\sigma_x\beta\sum_{y\in\tilde\Lambda_k^c} J_{xy}\omega_y+\frac{it\sigma_x}{\sqrt{D_k}}\biggr)\diff\sigma_x}{\displaystyle\int_{\mathbb R} \exp\biggl(-F(\sigma_x)+\sigma_x\beta\sum_{y\in\tilde\Lambda_k^c} J_{xy}\omega_y\biggr)\diff\sigma_x}\nonumber\\
&=&\frac{\sqrt{D_k}}{it}\frac{\displaystyle\int_{\mathbb R}
\Bigl(-F'(\sigma_x)+\beta\sum_{y\in\tilde\Lambda_k^c} J_{xy}\omega_y\Bigr)
\exp\biggl(-F(\sigma_x)+\sigma_x\beta\sum_{y\in\tilde\Lambda_k^c} J_{xy}\omega_y+\frac{it\sigma_x}{\sqrt{D_k}}\biggr)\diff\sigma_x}{\displaystyle\int_{\mathbb R} \exp\biggl(-F(\sigma_x)+\sigma_x\beta\sum_{y\in\tilde\Lambda_k^c} J_{xy}\omega_y\biggr)\diff\sigma_x}\
\end{eqnarray}
Hence
\be
\left|\nu_{\tilde{\Lambda}_k,\beta}^{\omega}\left(\exp\left(\frac{it\sigma_x}{\sqrt{D_k}}\right)\right)\right|\;\le\; \frac{\sqrt{D_k}}{\card t} \int_{\mathbb R} \Bigl(\card{F'(\sigma_x)} + \vartheta(x)\Bigr) \,\pi_+^\omega (\diff t)
\;=:\; \frac{\sqrt{D_k}}{\card t}\,\gamma\;.
\ee
For finite spins, $\gamma$ can be defined uniformly in $\omega$ by replacing $\pi^\omega_\pm$ by $\pi^\infty_\pm$.
\end{proof}

\section*{Acknowledgements}
The authors would like to acknowledge the support of the NYU-ECNU Institute of Mathematical Sciences at NYU Shanghai. E.O.E. thanks Research Institute for Mathematical Sciences, an International Joint Usage/Research Center located in Kyoto  University, for the support during NYU Shanghai-Kyoto-Waseda Young Probabilists' Meeting.

\appendix

\section{Some basic results on cluster expansion method}\label{sec:cluster-exp}
This section provides a summary of the cluster expansion theory for readers unfamiliar with this topic. The topic can be found in numerous references such as \cite{RN24,FP, FV}.
\subsection{Cluster expansion}\label{cluster_expansion_subsec}
In order to use the cluster expansion method, let us give the definition of \textit{polymers}, \textit{activities} and the compatible relation between the polymers.  

\begin{definition}\label{defi.1}
    A set $A\subset \Lambda$ is \empbf{connected} (for the \emph{in}compatibility relation) if it is non-empty and can not be partitioned into two mutually compatible subsets.
\end{definition}

\begin{definition}\label{defi.2}
    Finite connected sets are called \empbf{polymers}.  The set of polymers will be denoted $\mathcal{P}$.
\end{definition}

\begin{definition} 
\begin{itemize}
\item[(i)] The \empbf{support} of a family $\boldsymbol{B}=\{X_i\}_i$ of finite subsets of $\mathcal P^{1,2}$ is the set
$\underline{\boldsymbol{B}}=\cup_i X_i$.
\item[(ii)] Two finite families $\boldsymbol{B}_1,\boldsymbol{B}_2$ of finite subsets of $\Lambda$ are \emph{compatible} if their supports are disjoints i.e. $\underline{\boldsymbol{B}}_1 \cap \underline {\boldsymbol{B}}_2 = \emptyset$.
\item[(iii)] A family $\boldsymbol{B}$ is said \emph{connected} if it is not the union of two mutually compatible families.  
\end{itemize}
\end{definition}

Using the Mayer method to the partition function and the definition of the polymers in the Definition \ref{defi.1} and Definition \ref{defi.2}, we can write the partition function as the following form
\begin{equation}\label{defin.3}
    Z\;=\;1+\sum_{n\ge 1}\sum_{(R_1,\ldots, R_n)\in\mathcal{P}^n}\prod_{1\le i<j\le n}\mathbf{1}_{\underline{\boldsymbol{B}}_1\cap\underline{\boldsymbol{B}}_2=\emptyset}\prod_{i=1}^n\zeta(R)
\end{equation}
where $\zeta(R)$ is the activity of the polymer $R$.

\begin{definition} Consider a set $\mathcal P$ endowed with a hard-core compatibility relation ``$\sim$".  
\begin{itemize}
\item[(i)] The \empbf{incompatibility graph} of a sequence $(B_1,\ldots, B_n)\subset \mathcal P^{1,2}$ 
is the graph $\mathbb{G}(B_1,\ldots,B_n)$ with vertex set $\{1,\ldots,n\}$ and edge set 
\[\left\{\{i,j\}:B_i\nsim B_j,0\le i<j\le n\right\}.\]
\item[(ii)] The sequence $(B_1,\ldots, B_n)$ is a \empbf{cluster} if it is connected or, equivalently, if $\mathbb{G}(B_1,\ldots,B_n)$ is a connected graph.
\end{itemize}
\end{definition}

\begin{definition} The \empbf{cluster expansion} corresponding to the gas expansion \eqref{defin.3} is the formal power series in the activities $\zeta(R)$ obtained by taking the logarithm ---in the sense of composition of formal power series--- of the power series \eqref{defin.3}.
\end{definition}

\begin{theorem}\label{theo:cl1} The cluster expansion for the general gas expansion \eqref{defin.3} is
\begin{equation}\label{eq:polmect10.-1}
\log Z\;=\;\sum_{n=1}^{\infty}\frac{1}{n!}\sum_{\substack{(R_1,\ldots,R_n)\\ R_i\subset \mathcal P^{1,2}\\ B_i \;{\rm connected}}}\omega_n^T(R_1,\ldots,R_n)\,\zeta(R_1)\cdots \zeta(R_{n})
\end{equation}
with
\begin{equation}\label{eq:polmect12}
\omega_n^T(B_1,\ldots,B_n)\;=\;\left\{\begin{array}{cl}
1 & \text{if }\, n=1\\ [6pt]
\sum\limits_{\substack{G\subset\mathbb{G}(B_1,\ldots,B_n)\\G\, \mathrm{conn. spann.}}} (-1)^{|E(G)|} & \text{if }\,n>1 \text{ and }\mathbb{G}(B_1,\ldots,B_n)\text{ connected} \\ [6pt]
0 & \text{if } \mathbb{G}(B_1,\ldots,B_n)\, \text{ not connected}
\end{array}\right.,
\end{equation}
where $G$ ranges over all connected spanning subgraphs of $\mathbb{G}(B_1,\ldots,B_n)$. 
\end{theorem}

\subsection{Convergence criteria of cluster expansion}\label{conv.con}
Convergence is studied at the level of the \emph{pinned cluster expansion}
\begin{equation}\label{eq:polmect10.1.1}
\Pi(\bbd{\zeta})(B_0)\;=\;\sum_{n=0}^{\infty}\frac{1}{n!}\sum_{(B_1,\ldots,B_n)\in\mathcal{P}^n}\omega_n^T(B_0,B_1,\ldots,B_n)\,\zeta(B_1)\cdots \zeta(B_{n})
\end{equation}
which is controlled by the series obtained by termwise bounding by absolute values:
\begin{equation}\label{eq:polmect10.1.2.bis}
\card{\Pi}(\bbd{\lambda})(B_0)\;=\; \sum_{n=0}^{\infty}\frac{1}{n!}\sum_{(B_1,\ldots,B_n)\in\mathcal{P}^n}\card{\omega_n^T(B_0,B_1,\ldots,B_n)}\,\lambda_{\B_1}\cdots \lambda_{\B_{n}}
\end{equation}
with $\lambda_{\B_i}\in\mathbb{R}_+$.  As this last series has positive terms, its convergence amounts to the boundedness of the partial sums; a fact relatively simple to determine without recourse to Banach-space fixed points or other mathematically sophisticated techniques.   The convergence of \eqref{eq:polmect10.1.2.bis} implies the absolute convergence of \eqref{eq:polmect10.1.1} uniformly in the polydisc 
\begin{equation}
\mathcal{D}(\bbd\lambda)\;=\; \bigl\{\card{\zeta(B)}\le \lambda_B:B\in\pro\bigr\}\;.  
\end{equation}
The following theorem states the strongest convergence criterion available for the cluster expansion \eqref{eq:polmect10.1.1}.

\begin{theorem}\label{theo:rr.tan.06}
Consider the function $\boldsymbol{\varphi}:[0,+\infty)^{\mathcal{P}} \longrightarrow [0,+\infty]^{\mathcal{P}}$ defined by
\begin{equation}\label{eq:quan26}
\varphi_{B_0}(\boldsymbol{\mu})\;=\;
1+\sum_{n\ge1}\frac{1}{n!}\sum_{(B_1,\ldots,B_n)\in\pro^n}\mu_{B_1}\cdots \mu_{B_n}
\prod_{i=1}^n\mathbf{1}_{B_0\nsim B_i}\prod_{1\le k< \ell\le n}\mathbf {1}_{B_k\sim B_\ell}\;.
\end{equation}
If $\boldsymbol{\lambda}\in[0,+\infty)^{\mathcal{P}}$ satisfies 
\begin{equation}\label{eq:quan23.1}
\lambda_{B}\;\le\; \frac{\mu_{B}}{\varphi_{B}(\boldsymbol{\mu})}
\end{equation}
for each $B\in\mathcal{P}$, for some $\boldsymbol{\mu}\in[0,+\infty)^{\mathcal{P}}$, then 
the following holds uniformly in the polydisk $\mathcal{D}(\bbd\lambda):=\bigl\{\bbd\rho: \card{\rho_B}\le \lambda_B , B\in \pro\bigr\}$:
\begin{itemize}
\item[(a)] The pinned expansion \eqref{eq:polmect10.1.1} converge absolutely.

\item[(b)] For each $B\in\mathcal{P}$,
\begin{equation}\label{eq:mmm.1}
\card{\rho(B)}\,\card{\Pi(\bbd{\rho}(B))}\;\le\; \mu(B)\;.
\end{equation}
\end{itemize}

\end{theorem}
Theorem \ref{theo:rr.tan.06} yields the following slightly weaker criterion.

\begin{corollary}[Gruber-Kunz criterion]\label{corollary1.1}
A sufficient criterion for \eqref{eq:quan23.1} is the existence of a real number $a>0$ such that 
\begin{equation}\label{eq:rr.convergence.1}
\sup_{x\in \underline B_0}\sum_{\substack{B\in\pro\\x\in \underline B}}|\zeta(B)| e^{a|B|}\le e^{a}-1
\end{equation}
for each $B_0\in\pro$.
In this case, 
\begin{equation}\label{eq:rr.convergence.1.bis}
\card{\Pi}(\bbd{\zeta})(B)\;\le\; e^{a|B|}
\end{equation}
for each $B\in\pro$.
\end{corollary}

\subsection{Tree-graph identities}\label{tree_graphs}
The general setup involves a finite graph $\mathbb{G}=(\mathbb{U},\mathbb{E})$, the set $\mathcal{C}_{\mathbb{G}}$ of all connected spanning sub-graphs of $\mathbb{G}$ and the family $\mathcal{T}_{\mathbb{G}}$ of trees belonging to $\mathcal{C}_{\mathbb{G}}$. The set $\mathcal{C}_{\mathbb{G}}$ is partially ordered by bond inclusion:
\begin{equation}\label{eq:part.schem1}
G\le \tilde G\;\Longleftrightarrow\; E(G)\subset E(\tilde G).
\end{equation}
If $G\le \tilde G$, let us denote $[G,\tilde{G}]$ the set of $\widehat{G}\in\mathcal{C}_{\mathbb{G}}$ such that $G\le \widehat{G}\le \tilde G$. 

\begin{definition} A partition scheme of the family $\mathcal{C}_{\mathbb{G}}$ is a map $R:\mathcal{T}_{\mathbb{G}}\longrightarrow\mathcal{C}_{\mathbb{G}}:\tau\mapsto R(\tau)$ such that 
\begin{itemize}
\item[(i)] $E(\tau)\subset E(R(\tau))$, and
\item[(ii)] $\mathcal{C}_{\mathbb{G}}$ is the disjoint union of the sets $[\tau, R(\tau)] $, $\tau\in\mathcal{T}_{\mathbb{G}}$.
\end{itemize}
\end{definition}
The following Proposition is a straightforward consequence of  the definition of partition scheme. Let us recall our previous notation
\[
\omega^T_n(x_1,\ldots, x_n)\;=\;  \sum_{g\in\mathcal C[n]}\prod_{\{i,j\}\in E(g)}\bigl[w(x_i,x_j)-1\bigr].
\]

\begin{proposition}\label{prop:part.sch.1}
For any partition scheme $R$,
\begin{equation}\label{eq:part.sch.1}
\omega^T_n(x_1,\ldots, x_n)\;=\; \sum_{\tau\in\trees[n]}\prod_{\{i,j\}\in E(\tau)}\bigl[w(x_i,x_j)-1\bigr]\prod_{\{i,j\}\in E(R(\tau))\setminus E(\tau)}w(x_i,x_j)\;.
\end{equation}
\end{proposition}
\begin{proof}
The proof of this proposition is well known \cite{ScSo05}, but we reproduce it here for the sake of completeness.  Let us denote $f(i,j) :=w(x_i,x_j)-1$.  Then, 
\begin{eqnarray}
\sum_{\mathrm{g}\in\mathcal{C}[n]}\prod_{\{i,j\}\in E(g)}f(i,j)
&=&\sum_{\tau\in\trees[n]}\sum_{\mathrm{g}\in [\tau,R(\tau)]}\prod_{\{i,j\}\in E(g)}f(i,j)\nonumber\\
&=&\sum_{\tau\in\trees[n]}\prod_{\{i,j\}\in E(\tau)}f(i,j)
\sum_{\mathrm{g}\in [\tau,R(\tau)]}\prod_{\{i,j\}\in E(\mathrm{g})\setminus E(\tau)}f(i,j)\nonumber\\
&=&\sum_{\tau\in\trees[n]}\prod_{\{i,j\}\in E(\tau)}f(i,j)
\prod_{\{i,j\}\in E(R(\tau))\setminus E(\tau)}\bigl[f(i,j)+1\bigr]\nonumber\;.
\end{eqnarray}
\end{proof}

The usual way to prove that a map $R$ is indeed a partition scheme is by defining an inverse operation $T$ that maps intervals of subgraphs into a tree.  This strategy is formalized by the following well-known proposition.

\begin{proposition} \cite{ScSo05} \label{prop:rr.1}
The following statements are equivalent
\begin{itemize}
\item[(a)] There are two maps
\[\begin{tikzcd}[every arrow/.append style={shift left}]
  \mathcal{T}_{\mathbb{G}}\arrow{r}{R} & \mathcal{C}_{\mathbb{G}}\arrow{l}{T} 
 \end{tikzcd}\]
such that $T^{-1}(\tau)\;=\;\left\{G\in\mathcal{C}:\tau\le G\le R(\tau)\right\}$.
\item[(b)] $R$ is a partition scheme in $\mathcal{C}_{\mathbb{G}}$.
\end{itemize}
\end{proposition}
The proof of this proposition can be consulted, for instance, in \cite{ScSo05}.

Three types of partition schemes have been used to study convergence of cluster expansions: Penrose scheme \cite{FP,Pen67,ScSo05},  exploratory schemes proposed by Temmel \cite{Tem14} and the Kruskal partition introduced by Procacci and Yuhjtman \cite{PY17}.

\end{document}